\documentclass[11pt]{article}
\usepackage[colorlinks=true,linkcolor=blue,citecolor=Red]{hyperref}
\usepackage{amsmath}

\usepackage{mathtools}
\DeclarePairedDelimiter\ceil{\lceil}{\rceil}

\usepackage{bbm}

\usepackage{amsmath}
\usepackage{amssymb}
\usepackage{amsthm}
\usepackage{url}
\usepackage{bbm}

\usepackage{color}

\definecolor{Red}{rgb}{1,0,0}
\definecolor{Blue}{rgb}{0,0,1}
\definecolor{Olive}{rgb}{0.41,0.55,0.13}
\definecolor{Yarok}{rgb}{0,0.5,0}
\definecolor{Green}{rgb}{0,1,0}
\definecolor{MGreen}{rgb}{0,0.8,0}
\definecolor{DGreen}{rgb}{0,0.55,0}
\definecolor{Yellow}{rgb}{1,1,0}
\definecolor{Cyan}{rgb}{0,1,1}
\definecolor{Magenta}{rgb}{1,0,1}
\definecolor{Orange}{rgb}{1,.5,0}
\definecolor{Violet}{rgb}{.5,0,.5}
\definecolor{Purple}{rgb}{.75,0,.25}
\definecolor{Brown}{rgb}{.75,.5,.25}
\definecolor{Grey}{rgb}{.5,.5,.5}

\usepackage{url}
\newcommand{\ind}{\mathbbm{1}}

\newcommand{\G}{\mathbb{G}}

\renewcommand{\epsilon}{\varepsilon}

\newcommand{\EE}{\mathbb{E}}

\renewcommand{\G}{\mathbb{G}}

\newcommand{\ignore}[1]{\relax}

\newlength\myindent
\newtheorem{theorem}{Theorem}[section]
\newtheorem{remark}[theorem]{Remark}
\newtheorem{lemma}[theorem]{Lemma}
\newtheorem{conjecture}[theorem]{Conjecture}

\newtheorem{proposition}[theorem]{Proposition}
\newtheorem{coro}[theorem]{Corollary}

\newtheorem{definition}[theorem]{Definition}

\newcommand{\ER}{Erd{\H o}s-R\'{e}nyi }


\newcounter{parentnumber}

\makeatletter
\def\BState{\State\hskip-\ALG@thistlm}
\makeatother

\definecolor{Red}{rgb}{1,0,0}
\definecolor{Blue}{rgb}{0,0,1}
\definecolor{Olive}{rgb}{0.41,0.55,0.13}
\definecolor{Green}{rgb}{0,1,0}
\definecolor{MGreen}{rgb}{0,0.8,0}
\definecolor{DGreen}{rgb}{0,0.55,0}
\definecolor{Yellow}{rgb}{1,1,0}
\definecolor{Cyan}{rgb}{0,1,1}
\definecolor{Magenta}{rgb}{1,0,1}
\definecolor{Orange}{rgb}{1,.5,0}
\definecolor{Violet}{rgb}{.5,0,.5}
\definecolor{Purple}{rgb}{.75,0,.25}
\definecolor{Brown}{rgb}{.75,.5,.25}
\definecolor{Grey}{rgb}{.5,.5,.5}
\definecolor{Pink}{rgb}{1,0,1}
\definecolor{DBrown}{rgb}{.5,.34,.16}
\definecolor{Black}{rgb}{0,0,0}

\newcommand{\A}{\mathcal{A}}

\usepackage{float}

\usepackage{float}

\newcommand{\cA}{{\mathcal A}}
\newcommand{\cE}{{\mathcal E}}
\newcommand{\cF}{{\mathcal F}}
\newcommand{\cS}{{\mathcal S}}

\renewcommand{\Pr}{\mathbb{P}}

\newcommand{\indic}[1]{\mathbbm{1}_{\{{#1}\}}}

\newcommand\bigpar[1]{\bigl(#1\bigr)}
\newcommand\Bigpar[1]{\Bigl(#1\Bigr)}
\newcommand\biggpar[1]{\biggl(#1\biggr)}

\newcommand\bigsqpar[1]{\bigl[#1\bigr]}
\newcommand\Bigsqpar[1]{\Bigl[#1\Bigr]}
\newcommand\biggsqpar[1]{\biggl[#1\biggr]}

\newcommand\bigcpar[1]{\bigl\{#1\bigr\}}
\newcommand\Bigcpar[1]{\Bigl\{#1\Bigr\}}
\newcommand\biggcpar[1]{\biggl\{#1\biggr\}}

\newcommand\bigceil[1]{\bigl\lceil#1\bigr\rceil}

\newcommand\lrceil[1]{\left\lceil#1\right\rceil}
\newcommand\bigfloor[1]{\bigl\lfloor#1\bigr\rfloor}

\newcommand\lrfloor[1]{\left\lfloor#1\right\rfloor}

\renewcommand{\epsilon}{\varepsilon}
\newcommand{\eps}{\varepsilon}

\usepackage{algorithm}
\usepackage{algorithmic}

\newenvironment{romenumerate}[1][-5pt]{
\addtolength{\leftmargini}{#1}\begin{enumerate}
 }{\end{enumerate}}

\let\OLDthebibliography\thebibliography
\renewcommand\thebibliography[1]{
  \OLDthebibliography{#1}
  \setlength{\parskip}{0pt}
  \setlength{\itemsep}{0pt plus 0.3ex}
}

\title{Optimal Hardness of Online Algorithms for Large Independent Sets}
\author{\sf{David Gamarnik}\thanks{Massachusetts Institute of Technology; e-mail: {\tt gamarnik@mit.edu}} \and \sf{Eren C. K{\i}z{\i}lda\u{g}}\thanks{University of Illinois Urbana-Champaign; e-mail: 
 {\tt kizildag@illinois.edu}} \and \sf{Lutz Warnke}\thanks{University of California San Diego; e-mail: {\tt lwarnke@ucsd.edu} Supported by NSF~CAREER grant~DMS-2225631.}}
\date{January~25, 2026}
\begin{document}
\maketitle
\begin{abstract}
We study the algorithmic problem of finding a large independent set in the Erd{\H o}s-R\'{e}nyi  random graph $\G(n,p)$. In the dense regime (constant~$p$), the largest independent has size $2\log_b(n)$, while a simple greedy algorithm -- where vertices are revealed sequentially and each decisions depends only on previously seen vertices -- finds an independent set of size~$\log_b(n)$, where $b=1/(1-p)$. In his seminal 1976 paper, 
Karp challenged the community to either improve this guarantee or establish its hardness. Decades later, this remains one of the most prominent open problems in random~graphs. In the sparse regime, sharp lower bounds were obtained within the low-degree polynomial (LDP) framework via Overlap Gap Property (OGP) arguments tailored to stable algorithms. However, these techniques appear ineffective in the dense setting: LDP algorithms are conjectured to fail even in regimes where greedy succeeds, and we show that the greedy algorithm is unstable.

We consider a broad class of online algorithms with bounded access to future information, which includes the classical greedy as a special case. We show that for any $\epsilon>0$, no algorithm in this class finds an independent set of size $(1+\epsilon)\log_b(n)$ with non-trivial probability. Our lower bound holds uniformly for all edge-probabilities $p\in[d/n,1-n^{-1/d}]$, where $d$ is a large constant. For constant $p$, we further show that our result is nearly tight in the future-query budget by designing a quasipolynomial-time online algorithm that achieves $(1+\epsilon)\log_b(n)$ using slightly more queries.

Our results are thus tight both in the computational threshold $\log_b(n)$ and the future-query budget required to surpass it. In particular, they show that the  greedy algorithm is optimal among the online algorithms we consider, thus providing strong evidence for the hardness of Karp's problem. Our proof relies on a refined analysis of the geometry of large independent sets and establishes a novel variant of the OGP. While the OGP framework has been used primarily to rule out stable algorithms, online algorithms are inherently unstable, necessitating new technical ideas. We introduce several key ingredients, including: (i)~a temporal interpolation path that evolves with the algorithm, and (ii)~a carefully chosen stopping time that tracks when the algorithm's partial output reaches a critical size. We expect these ideas to be useful for extending OGP-based barriers to other online models.
\end{abstract}
\newpage
\section{Introduction}
In this paper, we study the algorithmic problem of finding a large independent set in the \ER random graph $\mathbb{G}(n,p)$, where each of the $\binom{n}{2}$ potential edges is present independently with probability~$p$. Throughout the introduction, we focus on the \emph{dense} regime where $p$ is a constant, though our main result holds for a broader range of~$p:=p(n)$.

In the worst-case, even approximating the largest independent set to within an $n^{1-\epsilon}$ factor is known to be NP-hard~\cite{hastad-clique,khot2001improved}. The situation however is drastically different and far more interesting in the presence of randomness. For $\G(n,p)$, the largest independent set has size approximately $2\log_b(n)$ with high probability (whp), where $b=b(p):=\frac{1}{1-p}$~\cite{matula1970complete,grimmett1975colouring,matula1976largest,bollobas1976cliques}. This is established via the moment method and is therefore non-constructive. On the other hand, a simple greedy algorithm --- making sequential decisions by only using the partial information available at a time --- finds an independent set of size approximately $\log_b(n)$ whp~\cite{grimmett1975colouring}. We describe this algorithm in greater detail below. In a seminal 1976 paper, Karp~\cite{karp1976probabilistic} asked whether it is possible to design a better polynomial-time algorithm that reaches $(1+\epsilon)\log_b(n)$ (whp) for an $\epsilon>0$, a tantalizing question that remains open decades later.

Karp's question is one of the most prominent algorithmic problems in random graphs and average-case complexity.\footnote{For example, in his 2014 ICM plenary lecture, Frieze~\cite{F2014} explicitly highlighted it as a challenging open problem that has surprisingly seen little progress.}
It is the oldest and arguably the most famous example of a \emph{statistical-computational gap} --- a gap between the optimal solution and the best value achieved by any known efficient algorithm. This gap has motivated Jerrum to introduce the infamous planted clique problem~\cite{jerrum1992large}, and analogous factor 2-gaps have since been discovered in other random graphs (see below), suggesting a universal phenomenon. For detailed discussions on these gaps, see the surveys~\cite{bandeira2018notes,gamarnik2021overlap,gamarnik2022disordered,gamarnik2025turing}. Despite extensive recent progress in understanding computational gaps in various algorithmic problems, Karp's problem has seen limited progress. In particular, neither an $\Omega(\log_b(n))$ algorithmic improvement, nor any evidence of the optimality of the greedy algorithm is known. As the greedy algorithm only uses partial information without accessing the entire graph, this is quite surprising, motivating us to ask:
\begin{center}
   \emph{Can we design a better algorithm or show that this is a hard problem?}
\end{center}

In this paper, we obtain the first sharp computational lower bounds for online algorithms for Karp's problem. More concretely, we establish that a broad class of online algorithms, which include the greedy as a special case, fail to find an independent set of size $(1+\epsilon)\log_b(n)$ (whp) for any $\epsilon>0$ (see Theorem~\ref{thm:Informal-Main}). Thus, the greedy algorithm is essentially optimal within this class. As we detail below, the best known algorithms for optimizing dense random graphs are online. As such, our results provide strong evidence towards the algorithmic hardness of Karp's problem. We next briefly highlight related random graph literature.

\paragraph{Sparse Random Graphs} The sparse case $\G(n,d/n)$ with constant $d$ has witnessed substantial progress. Mimicking the dense regime, it also exhibits an analogous factor-2 gap: the largest independent set has size approximately $2n \log(d)/d$~\cite{frieze1990independence,frieze1992independence,bayati2010combinatorial}, whereas the best known efficient algorithm achieves only $n \log(d)/d$~\cite{lauer2007large}, both whp in the double limit $n\to\infty$ followed by $d\to\infty$. This gap has led Gamarnik and Sudan~\cite{gamarnik2014limits,gamarnik2017} to introduce the \emph{Overlap Gap Property (OGP)} framework, through which sharp lower bounds were obtained for local algorithms~\cite{rahman2017local} and, more recently, for low-degree polynomial (LDP) algorithms~\cite{wein2020optimal}, a class capturing many popular algorithms including the one attaining~$n \log(d)/d$~\cite{lauer2007large,gamarnik2020low}. Similar sharp lower bounds were extended recently to sparse bipartite~\cite{perkins2024hardness} and multipartite~\cite{dhawan2024low} graphs.  As such, the sparse setting is relatively well-understood.

\paragraph{Dense Random Graphs} The situation differs markedly in the dense regime. While the existential and algorithmic guarantees for dense random graphs were obtained much earlier than the sparse, this is not the case for algorithmic lower bounds: formal hardness results for the dense regime are quite scarce. In particular, known hardness results regard only low-degree polynomial (LDP) algorithms~\cite{wein2020optimal,huang2025strong}, though it is unknown if the greedy algorithm~\cite{grimmett1975colouring,karp1976probabilistic} can be implemented as an LDP. In fact, in a recent American Institute of Mathematics workshop on low-degree methods~\cite{AIM2024}, it was conjectured that LDP algorithms perform much worse than the greedy on dense random graphs:
\begin{conjecture}\label{conj-main}
    On $\G(n,1/2)$, no degree $o(\log^2 n)$ polynomial finds an independent set of size $0.9\log_2 n$.
\end{conjecture}
\noindent
That is, while the LDP algorithms are highly effective across many high-dimensional problems including sparse random graphs (see~\cite{wein2025computational,kunisky2019notes} and references therein), they fail to find a half-optimal independent set in dense setting, a task attainable by simple greedy.
As such, LDP algorithms appear ineffective for dense random graphs. This highlights that the algorithmic paradigms for sparse and dense settings are quite different, where the online setting is particularly relevant in the latter. 

Additionally, prior work establishing algorithmic hardness in sparse random graphs through the OGP framework crucially leverage the stability properties of local and LDP algorithms. As we show in Proposition~\ref{Prop:Unstable} below, greedy algorithm (and online algorithms in general) may well in fact be unstable, making it difficult to adapt the techniques from sparse random graphs. For this reason, different techniques are required for dense random graphs.

Our proof substantially refines existing OGP-based arguments for algorithmic hardness, which exploit the geometry of large independent sets in~$\G(n,p)$. In particular, our refined approach introduces several new ideas that account for the algorithm's online nature, including (i) a temporal interpolation path that evolves with the algorithm, and (ii) a carefully chosen stopping time that tracks when the algorithm’s partial output reaches a certain critical size; we expect that these new ingredients will be useful for extending OGP-based barriers to other online models. 
For a technical overview, see Section~\ref{sec:TechOverview}.

\paragraph{Follow-up Work} As mentioned above, a fruitful line of work provided a thorough picture for the sparse case, both for $\G(n,d/n)$ as well as bipartite and hypergraph models. The progress however has been quite limited in the dense case. Our techniques fill this missing toolkit and have already proven useful beyond Karp's problem. In particular, leveraging our techniques, \cite{dhawan2025sharp} obtained sharp online lower bounds for dense bipartite random graphs\footnote{An even starker contrast between the dense and sparse regimes arises in dense random bipartite graphs: while a degree-1 polynomial succeeds in sparse case~\cite{perkins2024hardness},  a delicate online algorithm (beyond greedy) is required in the dense, see~\cite{dhawan2025sharp}.} 
analogous to sparse one~\cite{perkins2024hardness}, and a forthcoming work~\cite{hypergraph-wip} extends these to dense hypergraphs and multipartite models. Thus, for dense random graphs, our techniques help moving toward a more complete picture analogous to sparse setting.

\subsection{Online Algorithms: Informal Description}\label{sec:InformalOnline}
We next give an informal description of the class of online algorithms central to our results; see Definition~\ref{def:Operation} for formal statements. As a motivation, we begin with the classical greedy algorithm~\cite{grimmett1975colouring,karp1976probabilistic}, reproduced as Algorithm~\ref{Alg:Greedy} below. Starting from the empty set $I:=\varnothing$, the greedy algorithm sequentially considers the vertices $t\in\{1,\dots,n\}$ and adds $t$ to $I$ iff $t$ is not adjacent to any vertex $\ell \in I$, where $I$ is the current independent set. 

Note that at time $t$, the algorithm makes a decision for vertex $t$ using only the subgraph induced by the vertices $1,2,\dots,t$. In particular, it does not use any information about edges involving the remaining vertices. This captures the \emph{online} aspect: the decision at time~$t$ is based exclusively on the information revealed up to time $t$. For the classical greedy algorithm, the sole source of information is the edges within $\{1,\dots,t\}$, i.e., it has no access to future information. 
%


\paragraph{Future Queries} A natural question is  whether one can relax the online constraint.  What if the algorithm is allowed \emph{``look-ahead"} and can query the status of (some) \emph{future edges} at each step?\footnote{We call an edge as future if it is incident to a vertex not yet queried.} Does it still fail beyond the value $\log_b(n)$, or can it surpass this threshold? Our approach allows us to answer these questions under limited access to future information. Specifically, we consider a generalized class of online algorithms that, at step $t$, makes its decision based on the edges among vertices in $[t]$ together with a controlled set of future-edge queries $S_t$. In particular, we require that the total number of  queried future edges whose both endpoints lie in the final output independent set is at most $c\log^2_b(n)$, for a constant $c>0$. It turns out the $\log^2_b(n)$ scaling is nearly tight, as discussed~below. 

Our main result (Theorem~\ref{thm:MAIN}) shows that for sufficiently small $c$, this broader class of online algorithms also fails at the $\log_b(n)$ threshold: for any $\epsilon>0$, there exists a $c:=c(\epsilon)$ such that the independent set produced by any online algorithm satisfying properties above has size at most ${(1+\epsilon)\log_b(n)}$ whp. 
Note that the case $c=0$ corresponds to classical online algorithms with no look-ahead. Thus, our result identifies $\log_b(n)$ as a sharp computational threshold both for classical online algorithms and for the generalized class with limited look-ahead. Importantly, our result holds regardless of the computational complexity of the algorithm: its power is limited by the information it is allowed to access, rather than by its running time. 

What happens if the allowed future information \emph{budget} is greater, i.e., if $c$ is larger? Somewhat surprisingly, our results are nearly tight in the number of future-edge queries. Specifically, while online algorithms with sufficiently small $c>0$ cannot exceed the $\log_b(n)$ barrier, we do construct an online algorithm which, for every $\epsilon\in (0,1)$,
satisfies the above properties 
with $c=3\epsilon$ and outputs an independent set of size $(1+\epsilon)\log_b(n)$ (whp); see Theorem~\ref{thm:algorithm} for details. Consequently, $\log^2_b(n)$ is the right scaling for a phase transition in terms of future edge inspections: any (online) algorithm with an $O(\log^2_b(n))$ budget fails, whereas there exist algorithms that succeed with an $\Omega(\log^2_b(n))$ budget. In particular, one can surpass the $\log_b(n)$ threshold using limited look-ahead in a controlled way.

\paragraph{Limited Look-Ahead} Beginning with~\cite{feige2020finding}, a line of work~\cite{alweiss2021subgraph,racz2019finding,feige2021tight} has explored algorithms with limited look-ahead for the same problem. In these works, a limited number of edges are queried to guide decision-making; depending on the setting, one can construct independent sets of size $\alpha\log_b(n)$ with $\alpha>1$. A key distinction in our setup is that the decision for vertex~$t$ must be made at step~$t$, i.e., it cannot be deferred. This contrasts with the aforementioned papers, which follow a fundamentally different ``inspect first, decide later" paradigm.

\subsection{Statistical-Computational Gaps}\label{sec:SCGs}
We now review prior work on statistical-computational trade-offs analogous to Karp's problem. 
Random optimization problems --- optimization problems involving randomly generated inputs --- are central to disciplines ranging from computer science and probability to statistics and beyond. Arising from a diverse array of models, many such problems interestingly share a common feature: a \emph{statistical-computational gap}, namely a discrepancy between existential guarantees and what is algorithmically achievable. While the optimal value can often be identified using non-constructive techniques (e.g., the moment method), polynomial-time algorithms typically fall short of finding near-optimal solutions. 

Models exhibiting such gaps 
go beyond the problem considered here and include random constraint satisfaction problems~\cite{mezard2005clustering,achlioptas2006solution,achlioptas2008algorithmic,gamarnik2014limits,bresler2021algorithmic}, optimization over random graphs~\cite{gamarnik2014limits,rahman2017local,gamarnik2020low,wein2020optimal,perkins2024hardness}, binary perceptron and discrepancy minimization~\cite{gamarnik2023algorithmic,gamarnik2022algorithms,gamarnik2023geometric,kizildaug2023planted,li2024discrepancy}, spin glasses~\cite{chen2019suboptimality,gamarnikjagannath2021overlap,huang2021tight,gamarnik2023shattering,kizildaug2023sharp,huang2023algorithmic,sellke2025tight}, and more. For many of these models, all attempts at designing better algorithms have so far been unsuccessful, suggesting inherent algorithmic intractability.

Due to implicit randomness, standard worst-case complexity theory unfortunately offers little insight into such average-case hardness.\footnote{For a few notable exceptions where formal hardness can be established under standard complexity-theoretic assumptions such as $P\ne \#P$, see, e.g.,~\cite{ajtai1996generating,boix2021average,GK-SK-AAP,vafa2025symmetric}.} Nevertheless, several powerful frameworks have emerged to understand such gaps and give evidence of computational hardness  of random optimization problems, as well as other models, notably statistical problems involving planted structures. For a broader perspective on these approaches, see~\cite{wu2018statistical,bandeira2018notes,gamarnik2021overlap,gamarnik2022disordered,gamarnik2025turing}.

\subsection{Algorithmic Lower Bounds for Random Optimization Problems}\label{sec:AlgLBROP}
For random optimization problems, one of the most powerful approaches to establishing algorithmic lower bounds leverages intricate geometry of the optimization landscape. While not explicitly ruling out algorithms, a body of work~\cite{mezard2005clustering,achlioptas2006solution,achlioptas2008algorithmic} uncovered an intriguing connection in certain random computational problems: the onset of algorithmic hardness coincides with phase transitions in geometric features of the landscape. Gamarnik and Sudan~\cite{gamarnik2014limits} later initiated a formal framework known as \emph{Overlap Gap Property (OGP)}, which leverages these geometric properties to rigorously rule out broad classes of algorithms.\footnote{The term \emph{OGP} was formally introduced in~\cite{gamarnik2018finding}.} As we mentioned briefly, their focus was on the problem of finding a large independent set in sparse \ER random graph with average degree $d$. It is known that the largest independent set in this graph has (whp)\,size $2n \log(d)/d$ in the double limit $n\to\infty$ followed by $d\to\infty$~\cite{frieze1990independence,frieze1992independence,bayati2010combinatorial}. In contrast, the best known algorithm finds an independent set of size only $n \log(d)/d$~\cite{lauer2007large}, highlighting a factor-two gap mirroring our setting. Gamarnik and Sudan showed that independent sets of size $(1+1/\sqrt{2})n \log(d)/d$ either overlap substantially or are nearly disjoint, a structural property that leads to the failure of local algorithms and refuting a conjecture of Hatami, Lov{\'a}sz, and Szegedy~\cite{hatami2014limits}. Later work by Rahman and Vir{\'a}g~\cite{rahman2017local} established a sharp lower bound matching the algorithmic $n \log(d)/d$ threshold by considering the overlap pattern among many independent sets (multi-OGP). For the same algorithmic problem, lower bounds against low-degree polynomials (LDP) --- an emerging framework that capture implementations of many practical algorithms --- were first obtained in~\cite{gamarnik2020low} and sharpened in~\cite{wein2020optimal}. Since then, the OGP framework has found broad applicability beyond random graphs, including a variety of other random optimization~problems. 

For certain models, the OGP-based hardness result is supported with a classical average-case hardness guarantee. For example, a remarkable recent result~\cite{vafa2025symmetric} shows that beating the threshold predicted by OGP for the binary perceptron implies a polynomial-time algorithm for a worst-case instance of a variant of the shortest lattice vector problem, widely believed not solvable by polynomial time algorithms.
At the same time, some problems are in fact solvable by polynomial time algorithms beyond the OGP thresholds. In addition to the well known counterexample based on the random XOR-SAT instances which is solvable by Gaussian elimination, a more recent and interesting counterexample of this kind was constructed in \cite{li2024some} which shows that the shortest path problem, which is clearly solvable by polynomial-time algorithms, exhibits OGP in some random graph models. The shortest path problem can be solved by the linear programming technique and it is possible that problems admitting linear programming methods may 
overcome the OGP barrier.
We do not attempt to review the (rather extensive) literature on OGP here, and instead refer the reader to several above mentioned surveys, in particular~\cite{gamarnik2021overlap,gamarnik2022disordered,gamarnik2025turing} and the references therein.

\subsection{Prior Work: OGP as a Barrier to Stable Algorithms}\label{sec:prior}
At a high level, the OGP asserts that near-optimal tuples of solutions exhibiting intermediate levels of overlap do not exist. In the case of pairwise-OGP (as in~\cite{gamarnik2014limits}), this manifests as a dichotomy: any pair of solutions is either very close or very far apart, with no solutions in between. Combined with a judicious interpolation argument, this gap serves as a geometrical barrier for \emph{stable algorithms}. 
Informally, an algorithm $\A$ is stable if, for any inputs $\boldsymbol{A},\boldsymbol{A'}$ with a small $\|\boldsymbol{A}-\boldsymbol{A}'\|$, the corresponding outputs $\A(\boldsymbol{A})$ and $\A(\boldsymbol{A}')$ are close; see, e.g., \cite{gamarnik2022algorithms,gamarnik2023algorithmic} for a formal statement. Stable algorithms are a broad class capturing powerful frameworks such as local algorithms (factors of i.i.d.)~\cite{gamarnik2014limits,rahman2017local}, approximate message passing~\cite{gamarnikjagannath2021overlap}, low-degree polynomials~\cite{gamarnik2020low,bresler2021algorithmic} and low-depth Boolean circuits~\cite{gamarnik2021circuit}, among many. Crucially, when the solution space exhibits the OGP, these algorithms, by virtue of their smooth evolution, are unable to traverse the overlap gap.

\paragraph{Unstability of the Greedy Algorithm} The success of the OGP-based arguments crucially relies on stability, in particular the fact that stable algorithms cannot `jump across' regions of forbidden overlap. As such,  much of the OGP-based hardness results in the literature rule out stable algorithms modulo a few notable exceptions (see below). 

However, as we demonstrate next, the classical greedy algorithm~\cite{grimmett1975colouring,karp1976probabilistic} for finding an independent set of size approximately $\log_b(n)$ is unstable. Consequently, it is unclear how to establish algorithmic lower bounds in this case (e.g., via the OGP).
Given a graph $G=(V,E)$ with vertex set $V=\{1,\dots,n\}=:[n]$ and edge set~$E \subseteq \binom{V}{2}$, we describe the greedy algorithm following~\cite{grimmett1975colouring}:
\begin{algorithm}
\caption{Greedy Algorithm for Independent Sets}\label{Alg:Greedy}  
\begin{algorithmic}
\STATE \textbf{Input}: A graph $G=(V,E)$ with $V=\{1,\dots,n\}$.
\STATE \textbf{Initialize}: $I=\varnothing$
\STATE \textbf{for} $t=1,\dots,n$ {\bf do} \\
\quad Set $I\leftarrow I\cup \{t\}$ if  vertex~$t$ is in~$G$ not adjacent to any vertex~$\ell \in I$  
\STATE \textbf{end}
\STATE \textbf{Output}: $I$
\end{algorithmic}
\end{algorithm}\\
\noindent
As first shown by Grimmett and McDiarmid~\cite{grimmett1975colouring}, for any constant edge-probability~$p$ and any~${\eps>0}$, 
the above-described Algorithm~\ref{Alg:Greedy} run on the random graph $G\sim \G(n,p)$ returns an independent set of size at least~${(1-\eps)\log_b(n)}$~whp. 
We now record that Algorithm~\ref{Alg:Greedy} is not stable. 
\begin{proposition}[Informal]\label{Prop:Unstable}
Fix any edge-probability~$p \in (0,1)$. 
Then the classical greedy algorithm for independent sets (Algorithm~\ref{Alg:Greedy}) is unstable on~${G\sim \G(n,p)}$.
\end{proposition}
\noindent
We defer the proof sketch to Appendix~\ref{apx:Prop:Unstable}, since it is rather tangential to the main arguments of this~paper.

\subsection{Lower Bounds for Online Algorithms}
Despite being unstable, Algorithm~\ref{Alg:Greedy} exhibits an \emph{online} feature mentioned earlier: the decision at round~$t$ (whether $t\in I$) is based exclusively on the edges among the first $t$ vertices. This prompts the question: 
\begin{center}
   \emph{What are the fundamental limits of online algorithms in our context? }
\end{center}
As we noted above, while the OGP has been successfully used to obstruct stable algorithms for numerous random optimization problems, computational barriers beyond this class remain quite scarce. The idea of using the OGP to rule out online algorithms was first introduced in~\cite{gamarnik2023geometric} for discrepancy minimization, where the best known algorithm (by Bansal and Spencer~\cite{bansalspenceronline}) is online. A similar OGP-based obstruction to online algorithms was recently obtained for the graph alignment problem~\cite{du2025algorithmic}. Importantly though, the arguments of these papers do not extend to dense random graphs.\footnote{For a more restricted class of online algorithms (that is a strict subset of $0$-restricted online algorithms in the sense of Definition~\ref{def:Operation} and Condition~\eqref{eq:CondInformal}, but includes the classical online greedy algorithm), 
for constant edge-probabilities~${p \in (0,1)}$ variants of the arguments from~\cite{gamarnik2023geometric} appear to yield hardness of finding independent sets of size $(1+\eps)\log_b(n)$ only when ${\eps \ge \eps_0}:=\sqrt{5}-2\approx 0.236$, falling short of achieving the desired hardness for all~$\eps>0$. } 
Furthermore, both of these works consider the classical online setting, where the algorithm has no access to future information. 
This prompts the question: 
\begin{center}
   \emph{Is there an additional benefit to allowing even limited look-ahead?}
\end{center}

Our work addresses the two above-mentioned questions, and differs from~\cite{gamarnik2023geometric,du2025algorithmic} both conceptually and technically. Conceptually, we consider a richer class of online algorithms that can access limited future information. As discussed in Section~\ref{sec:InformalOnline}, while online algorithms without look-ahead fail to find independent sets of size $(1+\epsilon)\log_b(n)$ for any $\epsilon>0$ (Theorem~\ref{thm:MAIN}), those permitted a limited number of look-ahead queries can surpass this threshold (Theorem~\ref{thm:algorithm}). This highlights the added algorithmic power conferred by look-ahead access. Further, our arguments successfully capture that a phase transition, in terms of permitted amount of look-ahead, occurs precisely at scaling $\log_b^2(n)$.

\paragraph{Proof Techniques} Technically, our arguments for ruling out online algorithms uses several novel ingredients. 
In particular, we (i) construct temporally evolving interpolation paths tailored to the online setting, 
and (ii) define a stopping time that tracks the size of the algorithm’s output at each step. In particular, while virtually all prior OGP-based lower bound arguments are `algorithm-oblivious', the interpolation paths we design evolve with the algorithm, a salient feature that is necessary to address the online setting. We highlight that these ingredients are absent from previous OGP-based lower bound arguments; moreover, they appear necessary even for the classical online setting with no look-ahead. That is, even one is only interested in ruling out classical online algorithms (no look-ahead) at the threshold $\log_b(n)$, our arguments are necessary. Interestingly, these technical ideas also allow us to establish \emph{strong hardness}, ruling out even the algorithms succeeding with probability only $o(1)$. These collectively highlight the added technical power of these new ideas, which have already proven useful beyond Karp's problem and yielded sharp lower bounds for dense bipartite and hypergraphs~\cite{dhawan2025sharp,hypergraph-wip}. As such, our arguments are likely useful in analyzing online algorithms in other high-dimensional problems as well. For a high-level overview of these new technical~ingredients, see Section~\ref{sec:TechOverview}.


\subsection{Summary of Contributions}
Suppose that the vertices are probed in the order $v_1,\dots,v_t$. In the classical version, an algorithm is said to be \emph{online} if the decision for $v_t$ is based solely on the information revealed up to round $t$—the edges among the first $t$ vertices. 

Our focus is on a generalized class of online algorithms, which encompasses the classical greedy algorithm~\cite{grimmett1975colouring,karp1976probabilistic} as a special case. At each step $t$, the algorithm still inspects the connectivity among the first $t$ vertices, and additionally, a set of vertex-pairs $S_t$ which may include vertices not yet queried (future vertices). The decision for $v_t$ is then based on two sources:

At each step $t$, the algorithm may inspect not only the connectivity among the first $t$ vertices, but also query an additionally set of vertex-pairs $S_t$, which may include pairs involving future vertices. The decision for $v_t$ is then based on two sources:  (a) the subgraph induced by the vertices $\{v_1,\dots,v_t\}$ probed so far and (b) the edges among the pairs $E_t = S_1\cup\cdots \cup S_t$. For a formal description, see Definition~\ref{def:Operation}. 

\paragraph{Budget Constraint} As mentioned earlier, we consider online algorithms that access future information in a controlled way. Specifically, we say that an algorithm $\A$ is \emph{$c$-restricted} if 
\begin{equation}\label{eq:CondInformal}
    \left|E_n\cap \binom{\A(G)}{2}\right|\le c\log_b^2(np),
\end{equation}
where $\A(G)$ is the independent set constructed by $\A$ on graph $G$, and $\binom{\A(G)}{2}$ is the set of all vertex-pairs $\{i,j\}$ with~${i,j\in \A(G)}$. 
Intuitively speaking, Condition~\eqref{eq:CondInformal} only limits the number of future edge `look-ahead' queries among vertex-pairs that end up in the independent set constructed by the algorithm~$\A$
(this does not limit the total number of `look-ahead' queries, i.e., the algorithm is allowed to query many other future edges). 
%
Furthermore, Algorithm~\ref{Alg:Greedy} uses no look-ahead and thus satisfies~\eqref{eq:CondInformal} with~$c=0$. 

\paragraph{Sharpness of Computational Threshold} Complementing Algorithm~\ref{Alg:Greedy}, our first main result establishes a sharp algorithmic lower bound against $c$-restricted online algorithms, for sufficiently small~$c$.
\begin{theorem}[Informal version of Theorem~\ref{thm:MAIN}]\label{thm:Informal-Main} 
For any $\epsilon>0$ and sufficiently small~$c>0$, 
no $c$-restricted online algorithm finds an independent set in~$\mathbb{G}(n,p)$ of size at least~${(1+\epsilon)\log_b(n)}$~whp.
\end{theorem}
Together with the discussed guarantees for the greedy algorithm by Grimmett and McDiarmid~\cite{grimmett1975colouring} (see Section~\ref{sec:AlgLBROP}), 
Theorem~\ref{thm:Informal-Main} implies the following:
\begin{coro}\label{coro:nice}
The classical greedy algorithm for independent sets (Algorithm~\ref{Alg:Greedy}) is asymptotically optimal within the class of classical online algorithms with no look-ahead.
\end{coro}
Ever since Karp's seminal 1976 paper~\cite{karp1976probabilistic}, 
the greedy algorithm has remained the asymptotically best polynomial-time algorithm for finding independent sets in $\G(n,p)$. 
With this context in mind, Theorem~\ref{thm:Informal-Main} and Corollary~\ref{coro:nice} therefore provide rigorous evidence for the belief that $\log_b(np)$ might be the \emph{algorithmic threshold} for the problem of finding large independent sets in $\G(n,p)$. 

\paragraph{Sharpness of Budget Constraint} Our next main result concerns the sharpness of the bound~\eqref{eq:CondInformal} on the future edge queries, 
which we show to tight up to constants for dense random graphs (with constant edge-probability~$p$). 
\begin{theorem}[Informal version of  Theorem~\ref{thm:algorithm}]\label{thm:ALG-Informal}
For any $\epsilon \in (0,1)$, there exists a $3\epsilon$-restricted online algorithm that finds an independent set in~$\mathbb{G}(n,p)$ of size at least~${(1+\epsilon)\log_b(n)}$~whp.
\end{theorem}
In words, Theorem~\ref{thm:ALG-Informal} shows that the scaling $O(\log_b^2(np))$ is optimal up to constants. Of course, while the algorithm in Theorem~\ref{thm:ALG-Informal} is online, it runs in superpolynomial time (see Section~\ref{sec:sharpness}). We also emphasize that for any $\epsilon>0$, it is easy to construct a $3$-restricted algorithm --- e.g., by brute-force searching for a large independent set and then adding its vertices sequentially. Theorem~\ref{thm:ALG-Informal} strengthens this observation by showing that for constant edge-probabilities~$p$, one can in fact take $c\to 0$ as $\epsilon\to 0$.

\paragraph{Online Models} We emphasize that the online setting is far more than being a single algorithmic approach to solving a problem. It represents a fundamental computational model for decision-making under uncertainty, extensively studied in machine learning~\cite{rakhlin2010online,rakhlin2011online-a,rakhlin2011online-b,rakhlin2013online}, convex optimization~\cite{hazan2016introduction}, and beyond. As mentioned, the best known algorithms for optimization over dense random graphs are online. Furthermore, for several important models --- e.g., random k-SAT and combinatorial discrepancy theory~\cite{bansalspenceronline} --- the best known algorithms are either online or greedy, therefore naturally admitting online implementations. Given the central role that the online setting plays in modern data science, it is essential to understand its fundamental limitations. Our results contribute to this by extending the existing toolkit, with techniques that are potentially applicable to a broader range of problems.

\section{Setup and Main Results}

\subsection{Online Algorithms with Look-Ahead: Formal Definition} 
We focus on a generalized class of online algorithms that (i) access limited future information at each step, and (ii) may also use additional randomness—beyond the input graph $G$—through an independent seed random variable $\omega\in\Omega$.

%

\begin{definition}\label{def:Operation}
    Let $G\sim \G(n,p)$ have vertex set $[n]:=\{1,\dots,n\}$. A randomized algorithm $\A$, with internal randomness determined by the independent seed $\omega \in \Omega$, runs for $n$ rounds and keeps track of sets $V_t\subseteq [n]$, $\A_t(G) \subseteq V_{t}$ and $E_t\subseteq \binom{[n]}{2}$ (initially $V_0=\A_0(G)=E_0=\varnothing$). In each round~$t\in[n]$:
    \begin{enumerate}
        \item Based on $\omega$ and all information revealed so far, 
				$\A$ randomly selects a vertex $v_t\in[n]\setminus V_{t-1}$,  
				and reveals the edge-status  of all so-far unrevealed vertex-pairs~$\{v_t,v\}$ with~$v \in V_{t-1}$. 
        \item Based on $\omega$ and all information revealed so far,
				$\A$ selects a set $S_t\subseteq \binom{[n]}{2}$ of vertex pairs, 
				and reveals the edge-status of all so-far unrevealed vertex-pairs $\{u,v\}\in S_t$. 
        \item  Based on $\omega$ and all information revealed so far, 
				$\cA$ decides whether $v_t$ is added to its independent set or not, i.e., whether~${\A_t(G)= \A_{t-1}(G) \cup \{v_t\}}$ or~${\A_t(G)= \A_{t-1}(G)}$ holds (while ensuring that $\A_t(G)$ is an independent set), and also updates $V_t=V_{t-1}\cup \{v_t\}$ and $E_t = E_{t-1} \cup S_t$. 
    \end{enumerate}
At the end, we write $\A(G) = \A(G,\omega):= \A_n(G)$ for the final independent set constructed by the algorithm~$\cA$ (to avoid clutter, we usually omit the dependence on the seed~$\omega \in \Omega$ in our notation).
\end{definition}
Several remarks are in order. The seed $\omega\in \Omega$ captures the internal randomness of $\A$ and affects the information revealed at any round (as mentioned above, for brevity we often often omit the dependence on $\omega$ in our notation). 
Note that for each $t\in[n]$, the partial output $\A_t(G)$ at round $t$ is an independent set by construction. 
The set $V_t$ consists of all vertices that are inspected by $\A$ in the first $t$ rounds. 
The sets $S_t$ are allowed to contain vertex-pairs $\{i,j\}$ with at least one endpoint in $\{v_{t+1},\dots,v_n\}$, i.e., among the vertices not yet explored. 
For this reason, we refer to edges with endpoints in $S_t$ as \emph{future edges}. 
The set $E_t=S_1\cup \cdots \cup S_t\subseteq\binom{[n]}{2}$ consists of all such future edges revealed until the end of round~$t$. 

\subsection{Main Results}
Our focus is on the $c$-restricted online algorithms, i.e., those that simultaneously satisfy Definition~\ref{def:Operation} and Condition~\eqref{eq:CondInformal}. The latter conditions asserts that at most $c\log_b^2(np)$ inspected future edges are contained in the final output generated by the algorithm.




Our first main result shows that for small enough $c$, no $c$-restricted online algorithm can find an independent set of size $(1+\epsilon)\log_b(np)$ with non-negligible probability. Formally:
\begin{theorem}\label{thm:MAIN}
For any~$\epsilon>0$, there exist constants $c,\xi>0$ and $d,n_0 >1$ (which may depend on~$\epsilon)$ 
such that, for all $n \ge n_0$, any $c$-restricted online algorithm $\A$ satisfies
\begin{equation}\label{eq:thm:MAIN}
\mathbb{P}_{\displaystyle G\sim \mathbb{G}(n,p),\omega \in \Omega}\Bigl[\bigl|\A(G,\omega)\bigr|\ge (1+\epsilon)\log_b(np)\Bigr] \le (np)^{-\xi \log_b(np)}
\end{equation}
for all edge-probabilities $p=p(n)$ satisfying $d/n\le p \le 1-n^{-1/d}$. 
\end{theorem}
Our proof of Theorem~\ref{thm:MAIN} develops a new toolkit based on the Overlap Gap Property. See Section~\ref{sec:TechOverview} for a technical overview, and Sections~\ref{sec:proof:main}-\ref{pf:MAIN-OGP} for the full proof. Our arguments show that it suffices to take
\[
c=\min\{\epsilon^2,1\}/8\qquad\text{and}\qquad \xi = \min\{\epsilon^2,1\}/64.
\]
\noindent Several remarks are in order. Firstly, Theorem~\ref{thm:MAIN} establishes that online algorithms cannot surpass the  threshold $\log_b(np)$. Combined with Conjecture~\ref{conj-main}, our lower bound therefore gives strong evidence towards the conjecture that the location of the computational threshold is at $\log_b(np)$.

Secondly, note that the probability bound~\eqref{eq:thm:MAIN} always satisfies
\[
(np)^{-\xi \log_b (np)} \le n^{-\Omega(1)}\to 0 \qquad \text{ as $n\to\infty$,} 
\]
since $\log_b(np) = \Omega(\log(d)/p)$, see Lemma~\ref{lemma:Logb-LowerBd} for details. 
Thus, Theorem~\ref{thm:MAIN} rules out algorithm succeeding even with vanishing probability, thereby establishing what is known as \emph{strong hardness}. See~\cite{huang2025strong} for strong hardness results regarding low-degree polynomials in other random computational~problems.

Importantly, the probability bound~\eqref{eq:thm:MAIN} is optimal, up to the numerical value of the constant~$\xi$ in the exponent: 
indeed, the deterministic $0$-restricted algorithm which sequentially inspects the the first ${s:=\ceil{(1+\epsilon)\log_b(np)}}$ vertices of  $\mathbb{G}(n,p)$ finds an independent set of size~$s$ with probability \[
\displaystyle (1-p)^{\binom{s}{2}} = (np)^{-\Theta\bigl((1+\epsilon)^2\bigr) \cdot  \log_b(np)}.
\]

\subsection{Sharpness of the Budget Constraint}\label{sec:sharpness}
Specializing on the dense regime where $p$ is constant, we now state our next main result regarding the sharpness of our bound on $c$-restrictedness in~\eqref{eq:CondInformal}. 

%
%
\begin{theorem}\label{thm:algorithm}
Fix any edge-probability~$p \in (0,1)$. Then, for any~$\epsilon \in (0,1)$, there exists a deterministic $3\eps$-restricted online algorithm~$\A$ 
 that satisfies, as~$n \to \infty$, 
\begin{equation}\label{eq:thm:algorithm}
\mathbb{P}_{\displaystyle G\sim \mathbb{G}(n,p)}\Bigl[\bigl|\A(G)\bigr|\ge (1+\epsilon)\log_b(n)\Bigr]=1-o(1). 
\end{equation}
\end{theorem}
Note that since the algorithm in Theorem~\ref{thm:algorithm} is deterministic, the probability is taken with respect to $G\sim \mathbb{G}(n,p)$ only (cf.~\eqref{eq:thm:MAIN} above). 
Our algorithm~$\A$  has superpolynomial running time $n^{(2+o(1))\eps^2 \log_b(n)}$; see Section~\ref{sec:PF-Greedy-BruteForce} for further~details.

Importantly, Theorem~\ref{thm:algorithm} asserts that our algorithmic lower bound in Theorem~\ref{thm:MAIN} is tight not only in terms of the threshold $\log_b(n)$, but also in the number\footnote{The analysis in Section~\ref{sec:PF-Greedy-BruteForce} reveals that Algorithm~$\A$ from Theorem~\ref{thm:algorithm} queries a total of $n^{\min\{1,2\eps\}+o(1)}$ many future edges (i.e., significantly more than those in the independent set), 
which demonstrates that it can be advantageous for online algorithms to query additional future edges outside of the final independent set.} 
of future edge queries. That is:
(i) for any $\epsilon>0$ and sufficiently small $c$, online algorithms fail to return an independent set of size $(1+\epsilon)\log_b(np)$; whereas (ii) for constant $p$, there exists a $3\epsilon$-restricted online algorithm that successfully returns an independent set of this size. Consequently, the scaling $\log_b^2(n)$ arising
in~\eqref{eq:CondInformal} is sharp.
%

\section{Overview of Technical Arguments}\label{sec:TechOverview}
Our algorithmic hardness proof exploits the geometry of large independent sets. Specifically, we develop a new variant of the Overlap Gap Property (OGP) tailored to the online setting. 

We first describe the classical OGP arguments for establishing hardness, which rely on two pillars. First, one shows that certain $m$-tuples of near-optimal solutions across correlated instances (called \emph{interpolation paths}) do not exist—regardless of the algorithm. Second, one proves that such $m$-tuples must be accessible to successful algorithms, leading to a contradiction.
  
  For random graphs, a canonical interpolation path consists of a sequence $G^{(1)}, G^{(2)}, \cdots ,G^{(\binom{n}{2})}$ of correlated copies of $\mathbb{G}(n,p)$, where, for each $i$, $G^{(i+1)}$ is obtained from $G^{(i)}$ by resampling a single edge. Thus, $G^{(1)}$ and $G^{(\binom{n}{2})}$ are independent, whereas the successive pairs $G^{(i)}, G^{(i+1)}$ are highly correlated. The argument exploits this correlation along with algorithmic stability. For certain other models, e.g., discrepancy minimization and binary perceptron~\cite{gamarnik2023geometric}, the interpolation involves random matrices where a certain fraction of columns are shared across instances and the remaining columns are resampled~independently.

The situation however is different in our setting. To begin with, even for the classical online algorithms (Algorithm~\ref{Alg:Greedy}), standard arguments fail to establish hardness at the $(1+\epsilon)\log_b(np)$ threshold for all~$\epsilon>0$. To obtain sharp bounds, we introduce substantial technical refinements.

\paragraph{Interpolation Paths with Temporal Structure} Let $G_1\sim \mathbb{G}(n,p)$. One of our new technical ingredients is to endow the interpolation path with a \emph{temporal} structure that evolves as the algorithm runs on the random graph~$G_1$. For suitable~$m \ge 1$, at each step $1 \le T \le n$, we construct $m$ correlated copies $G_1^{(T)},\ldots G_m^{(T)}$ of $\mathbb{G}(n,p)$ based on the available information (all with the same vertex set). 

For concreteness, let $\mathcal{F}_T$ be the set of vertex-pairs whose edge-status has been revealed by step~$T$. 
If there is no look-ahead, then~$\mathcal{F}_T$ contains all vertex-pairs $\{i,j\}$ among the inspected vertices $V_T$, i.e., $\mathcal{F}_T = \binom{V_T}{2}$. 
In the general case, $\mathcal{F}_T = \binom{V_T}{2}\cup E_T$, where $E_T:= S_1\cup\cdots \cup S_T$ is the set of future vertex-pairs queried by the algorithm until step~$T$. 
Notice that both $V_T$ and $E_T$ depend on the algorithm~$\A$ as well as the graph~$G_1$.


We now construct the correlated graphs $G_1^{(T)},\ldots G_m^{(T)}$.  
For~$i=1$ we set $G_1^{(T)}=G_1$, which means that the first copy is identical to the original random graph $G_1\sim \mathbb{G}(n,p)$. 
For each~$2 \le i \le m$, the edge-status of all vertex-pairs $\{u,v\} \in \mathcal{F}_T$ revealed during the first $T$ rounds is the same in $G_1$ and $G_{i}^{T}$, 
and all other vertex-pairs $\{u,v\} \not\in \mathcal{F}_T$ are included as an edge into $G_i^{(T)}$ independently with probability~$p$. 


This construction ensures that at each step $T$, the edge-status of all so-far revealed vertex-pairs are shared across the $m$ correlated copies, while the edge-status of all so-far unrevealed edges remain independent. Since $\mathcal{F}_T$ depends on algorithm's trajectory (see Definition~\ref{def:Operation}), the interpolation path is \emph{dynamic}, evolving as the algorithm progresses. In contrast, prior OGP-based barriers are oblivious to the algorithm and lack this temporal feature. Such adaptivity is essential in our setting, due to the online nature of the algorithm $\A$.

\paragraph{Stopping Time} Another crucial technical ingredient is a carefully defined stopping time $\tau$, which marks the first time the algorithm’s partial output on~$G_1$ reaches a critical size. Recall from Definition~\ref{def:Operation} that $\mathcal{A}_t(G_1)$, the algorithm's output after~$t$ steps, is an independent set for all~$t$. 
Ignoring some technicalities, we intuitively define $\tau$ as the first step $t \ge 1$ where $|\A_t(G_1)|$ has size $N:=(1-\epsilon/2)\log_b(np)$.

Importantly, $\tau$ has two source of randomness: the graph $G_1$, as well as the randomness of $\A$. Now define $I_i$ as the output of $\A$ on the correlated copy $G_i^{(\tau)}$ for $1\le i\le m$. Our argument leverages the geometric properties of the $m$-tuples of independent sets $I_1,\dots,I_m$. Specifically, we prove that there exists no $m$-tuple $(I_1,\dots,I_m)$ such that:%
{\vspace{-0.25em}%
\begin{romenumerate}
\itemsep 0em \parskip 0.125em  \partopsep=0.125em \parsep 0.125em  
    \item For each $1\le i\le m$, $I_i$ is an independent set in $G_i^{(\tau)}$ with $|I_i|\ge (1+\epsilon)\log_b (np)$,
    \item For every $1\le i\le m$, $I_i\cap V_T = I_1\cap V_T$ with $|I_i\cap V_T|=N$,
		\item $|\binom{I_i}{2}\cap E_T|\le c\log_b^2(np)$.\vspace{-0.125em}
\end{romenumerate}}%
\noindent 
For technical details and relevant structure, see Definition~\ref{def:Forbidden-Pattern}. To the best of our knowledge, this is the first OGP barrier where a stopping time argument appears necessary. 


These two components—adaptive interpolation and algorithm-dependent stopping time—are essential: they are required even for obstructing the classical greedy algorithm with no look-ahead. Without these arguments, it appears not possible to establish algorithmic hardness for all $\epsilon>0$. Notably, the careful design of the stopping time also enables us to rule out online algorithms that succeed with only vanishing probability, thereby strengthening our lower bound (as discussed further below).
\paragraph{Strong Hardness} A final key aspect of our approach is \emph{strong hardness}. Classical OGP arguments typically rule out algorithms that succeed with high probability $1-o(1)$ by controlling the \emph{success event}, i.e., the probability that an algorithm outputs near-optimal solutions across all correlated copies. This is usually done via union bounds and does not preclude algorithms that succeed with small or vanishing probability.

Recent work~\cite{huang2025strong} develops a framework for proving strong hardness—i.e., ruling out algorithms that succeed with vanishing probability $o(1)$—but their techniques are tailored to stable algorithms, rather than the online setting we consider. 

To establish strong hardness in our context, we design a modified success event—based on stopping time $\tau$—and apply several careful layers of Jensen’s inequality. This allows us to bound the success probability of \emph{any} online algorithm by $\exp(-\Omega(\log^2(n))$. While~\cite{gamarnik2023geometric} also applies Jensen's inequality in an online context and prove strong hardness, their argument is much simpler, essentially conditioning on common randomness and applying the tower property. Furthermore, that argument breaks down in our setting: it only rules out algorithms with success probability $1-o(1)$. In contrast, our refined approach yields a stronger guarantee. See Lemma~\ref{lem:comparison} for technical details. We believe these techniques may be broadly useful for establishing strong hardness in other online or adaptive algorithmic models.

\section{Impossibility Result: Proof of Theorem~\ref{thm:MAIN}}\label{sec:proof:main}
In this section, we prove our main impossibility result, Theorem~\ref{thm:MAIN}, 
which asserts that online algorithms operating according to Definition~\ref{def:Operation} and satisfying~\eqref{eq:CondInformal} cannot find independent sets of size $(1 + \eps)\log_b(np)$ with non-negligible probability. 

To this end, we may without loss of generality assume that $\eps \le 1$ (since the desired estimate~\eqref{eq:thm:MAIN} for~$\eps \ge 1$ is implied by estimate~\eqref{eq:thm:MAIN} for~$\eps=1$). 
Given~$\eps \in (0,1]$, with foresight we define 
\begin{equation}\label{eq:p-c:m}
c := \frac{\eps^2}{8},
\qquad 
\xi := \frac{\eps^2}{64}, 
\quad \text{ and } \quad 
m := \lrceil{\frac{16}{\epsilon^2}} , 
\end{equation}
and defer the choice of the constants $d,n_0 >1$.
We henceforth fix a $c$-restricted algorithm $\A$ per~\eqref{eq:CondInformal}.   

To avoid clutter, we henceforth always tacitly assume that $n$ is sufficiently large whenever necessary.\footnote{How large may depend on all other parameters $\eps,c,\xi,m,d$.} 
The parameter~$\log_b(np)$ is a bit cumbersome to work with, since in different parts of the edge-probability range $d/n\le p \le 1-n^{-1/d}$, it behaves quite differently. For this reason, we now record a convenient lower bound based on standard asymptotic estimates.
\begin{lemma}\label{lemma:Logb-LowerBd}
    There exists a universal constant $c_0>0$ (which does not depend on $\eps,c,\xi,m,d,n_0$) such~that 
\begin{equation}\label{eq:logbnp:estimate}
\log_b(np) 
= \frac{\log(np)}{-\log(1-p)} 
\ge \frac{c_0 \log(d)}{p} .
\end{equation}
\end{lemma}
The proof of Lemma~\ref{lemma:Logb-LowerBd} is routine, and the details are deferred to Appendix~\ref{apx:lemma:Logb-LowerBd}.

\subsection{Reduction to deterministic algorithms} 
Note that the algorithm $\A$ in Theorem~\ref{thm:MAIN} is randomized: 
we shall use a standard averaging-based reduction to show that in our upcoming arguments it suffices to consider a deterministic algorithm. 
\begin{lemma}\label{lem:deterministic}
There exists~${\omega^*\in \Omega}$ such that the deterministic algorithm ${\A(G):=\A(G,\omega^*)}$ satisfies 
\begin{equation}\label{eq:deterministic-A}
\mathbb{P}_{G\sim \G(n,p),\omega \in \Omega}\bigl[|\A(G,\omega)|\ge (1+\epsilon)\log_b(np)\bigr] \le \mathbb{P}_{G\sim \G(n,p)}\bigl[|\A(G)|\ge (1+\epsilon)\log_b(np)\bigr] .
\end{equation}
\end{lemma}
\begin{proof}
By explicitly averaging over the randomness $\omega \in \Omega$, we can~write 
\[
\mathbb{P}_{G\sim \G(n,p),\omega\in \Omega}\bigl[|\A(G,\omega)|\ge (1+\epsilon)\log_b(np)\bigr] = \mathbb{E}_{\omega \in \Omega}\Bigl[\mathbb{P}_{G\sim \G(n,p)}\bigl[|\A(G,\omega)|\ge (1+\epsilon)\log_b(np)\bigr]\Bigr] .
\]
Hence there exists a choice of~${\omega^*\in \Omega}$ such that ${\A(G)=\A(G,\omega^*)}$ satisfies estimate~\eqref{eq:deterministic-A}.
\end{proof}

We henceforth fix a choice $\omega^*\in \Omega$ that is guaranteed by Lemma~\ref{lem:deterministic}, and set $\A(G):=\A(G,\omega^*)$. 
In the remainder of the proof of Theorem~\ref{thm:MAIN}, 
our goal will be to derive an upper bound on the success probability of the deterministic algorithm $\A$ (which suffices by estimate~\eqref{eq:deterministic-A} from Lemma~\ref{lem:deterministic}). 
Conceptually, the main advantage of working with a deterministic algorithm~$\cA$ is that the only source of randomness is the graph on which~$\cA$ operates.

\subsection{Interpolation paths: correlated random graphs}\label{sec:interpolation}
The strategy of our proof of Theorem~\ref{thm:MAIN} is to analyze the behavior of the algorithm~$\cA$ on multiple correlated random graph families, referred to as interpolation paths.  
To this end we define 
\begin{equation}\label{eq:G_1}
    G_1\sim \mathbb{G}(n,p)
\end{equation}
as an  independent copy of the binomial random graph $\G(n,p)$, i.e., where each possible pair of vertices $\{u,v\} \in \tbinom{[n]}{2}$ is included as an edge independently with probability~$p$. 
We shall run the algorithm~$\cA$ on~$G_1$ to dynamically construct the families
\begin{equation*}\label{eq:Phi-T}
G_1^{(T)},\dots,G_m^{(T)} \qquad \text{for $1 \le T \le n$}
\end{equation*}
of correlated random graphs as follows, where to avoid clutter we henceforth write 
\begin{equation}\label{def:VT:ET}
V_T=V_T(G_1)=\{v_1, \ldots, v_T\} \qquad \text{ and } \qquad  E_T=E_T(G_1) = S_1\cup \cdots \cup S_T .
\end{equation}
For the first graph sequence~$G_1^{(T)}$ we simply set 
\begin{equation}\label{eq:G_1T}
G^{(T)}_{1}:=G_1 \qquad \text{for all~$1 \le T \le n$,}
\end{equation}
and the basic idea for the other graphs $G_i^{(T)}$ with~$1 \le T \le n$ is heuristically speaking as follows: 
the edge-status of all vertex-pairs revealed by~$\cA$ during the first~$T$ rounds is the same in~$G_1$ and~$G_i^{(T)}$, 
and all other vertex-pairs are included as an edge into~$G_i^{(T)}$ independently with probability~$p$. 
More formally, denoting by~$e_{i,T}\bigpar{\{u,v\}} \in \{0,1\}$ the edge-status of the pair~$\{u,v\}$ in~$G_i^{(T)}$, 
we~define
\begin{align}
\label{def:GiT:early}
e_{i,T}\bigpar{\{u,v\}} := e_{1,T}\bigpar{\{u,v\}} \qquad & \text{for all~$1 \le T \le n$, $2 \le i \le T$ and $\{u,v\} \in \tbinom{V_T}{2} \cup E_T$,}\\
\label{def:GiT:late}
e_{i,T}\bigpar{\{u,v\}} \stackrel{\text{i.i.d.}}{\sim} {\rm Bernoulli}(p) \qquad & \text{for all~$1 \le T \le n$, $2 \le i \le T$ and $\{u,v\} \not\in \tbinom{V_T}{_2} \cup E_T$.}
\end{align}

We now record some basic properties of the correlated random graphs $G_1^{(T)},\dots,G_m^{(T)}$.
%
\begin{remark}\label{rem:marginals}
For any $1\le i\le m$ and $1\le T\le n$ we have $G_i^{(T)} \sim \G(n,p)$, 
since definitions~\eqref{eq:G_1T} and \eqref{def:GiT:early}--\eqref{def:GiT:late} ensure that each potential edge of~$G^{(T)}_{i}$ is included independently with probability~$p$. 
\end{remark}
\begin{remark}\label{rem:run}
For any~$1 \le T \le n$, the first~$T$ rounds of the algorithm~$\cA$ on each graph $G_1^{(T)},\dots,G_m^{(T)}$ behave exactly the same as on the graph~$G_1$, 
since definitions~\eqref{def:GiT:early} and~\eqref{eq:G_1T} ensure that all vertex-pairs revealed during those rounds have exactly the same edge-status as in~$G^{(T)}_{1}=G_1$ 
(which by Definition~\ref{def:Operation} is all the information that the deterministic algorithm $\A$ uses to make its selections and decisions).  
\end{remark}

\subsection{Comparison argument via successful event~$\cS$}\label{sec:success}
Our proof of Theorem~\ref{thm:MAIN} hinges on the careful definition of the `successful' event $\cS$ in~\eqref{eq:def:S} below, 
which captures the idea that the algorithm~$\cA$ finds an independent set of size at least $(1+\eps) \log_b(np)$ in the correlated random graphs from Section~\ref{eq:Phi-T}. 
Turning to the details, setting with foresight 
\begin{equation}\label{eq:TAU}
\tau : = \min\Bigcpar{1 \le t \le n: \: |\A_t(G_1)| = \bigceil{(1-\epsilon/2)\log_b(np)} \text{ or } \: t=n}, 
\end{equation}
we define the `successful' event 
\begin{equation}\label{eq:def:S}
\cS := \bigcap_{1\le i\le m}\Bigcpar{\bigl|\cA(G_i^{(\tau)})\bigr|\ge (1+\epsilon)\log_b(np)} = \bigcap_{1\le i\le m} \cE_{i,\tau},
\end{equation}
where~$\cE_{i,T}$ denotes the event that ${\bigl|\cA(G_i^{(T)})\bigr|} \ge {(1+\epsilon)\log_b(np)}$. 
The strategy of our comparison-based argument is to estimate the probability $\Pr[\cS]$ in two different~ways (from below and~above), 
which eventually establishes the desired probability bound~\eqref{eq:thm:MAIN} of our main result Theorem~\ref{thm:MAIN}; see Section~\ref{sec:put:together}.
\begin{remark}
Definitions~\eqref{eq:TAU} and~\eqref{eq:def:S} of the stopping time~$\tau$ and the event~$\cS$ are an important ingredient of our proof. 
Specifically, in the proof of Lemma~\ref{lem:comparison} below they are key for avoiding the need to take a union bound over all possible~$1 \le T \le n$ where $|\A_T(G_1)| = \lrceil{(1-\epsilon/2)\log_b(np)}$ might occur. This is crucial for showing that the
resulting probability bound~\eqref{eq:thm:MAIN} goes to~zero in our main result Theorem~\ref{thm:MAIN}, i.e., to establish {\bf strong hardness} guarantees.
\end{remark}

\subsubsection{Lower bound on $\Pr[\cS]$: conditioning and convexity}\label{sec:success:lower}
The following lemma estimates the `success' probability $\Pr[\cS]$ from below, 
by relating it to the probability~$\Pr[\cE]$ that the algorithm~$\cA$ finds an independent set of size at least $(1+\eps) \log_b(np)$ in~$G_1 \sim \G(n,p)$. 
Concretely, the below proof of~\eqref{eq:E-s-A1-m} combines the careful definitions~\eqref{eq:TAU} and~\eqref{eq:def:S} of~$\tau$ and~$\cS$ 
with an application of Jensen's inequality and the fact that, after suitable conditioning, certain events become conditionally independent (due to the way the algorithm~$\cA$ makes its decisions, see~Definition~\ref{def:Operation}). 
\begin{lemma}\label{lem:comparison}
Denoting by~$\cE$ the event that ${\bigl|\cA(G_1)\bigr|} \ge {(1+\epsilon)\log_b(np)}$, 
we~have 
\begin{equation}\label{eq:E-s-A1-m}
\Pr[\cS]  \ge \Pr[\cE]^m .
\end{equation}
\end{lemma}
\begin{proof}
Let~$\Xi_T$ encode the edge-status of all vertex-pairs~$\{u,v\} \in \tbinom{V_T}{2} \cup E_T$, 
i.e., the edge-status of all vertex-pairs revealed during the first~$T$ rounds of the algorithm~$\cA$ 
(which intuitively corresponds to the common randomness of $G^{(T)}_1, \ldots, G^{(T)}_m$). 
We have
\begin{equation}\label{eq:Condition-1} 
\begin{split}
    \mathbb{P}[\tau=T, \; \cS] &= \mathbb{P}[\tau=T, \; \cE_{1,T},\dots,\cE_{m,T}] \\
&=\mathbb{E}\bigsqpar{\mathbb{P}[\tau=T, \; \cE_{1,T},\dots,\cE_{m,T} \: | \: \Xi_T]} \\
&=\mathbb{E}\bigsqpar{\ind\{\tau=T\}\Pr[\cE_{1,T},\dots,\cE_{m,T} \: | \: \Xi_T]},
\end{split}
\end{equation}
where the last equality in~\eqref{eq:Condition-1} follows from the fact that the event $\tau=T$ is determined by $\Xi_T$. 
Recall that~$\cE_{i,T}$ denotes the event that ${\bigl|\cA(G_i^{(T)})\bigr|} \ge {(1+\epsilon)\log_b(np)}$, 
where $\A$ is a deterministic algorithm that runs on the random graph~$G_i^{(T)}$.  
Note that by~\eqref{def:GiT:late}, conditional on the edge-status $\Xi_T$ of the first~$T$ rounds,  all remaining pairs~$\{u,v\} \not\in \tbinom{V_T}{2} \cup E_T$ are included as edges into each graph $G_1^{(T)}, \ldots, G_m^{(T)}$ independently with probability~$p$.  
Combined with Remark~\ref{rem:run} and $G_1^{(T)}=G_1$, it thus follows that 
\begin{align}\label{eq:independence}
\Pr[\cE_{1,T},\dots,\cE_{m,T} \: | \: \Xi_T] 
& = \prod_{1 \le i \le m}\Pr[\cE_{i,T} \: \big| \: \Xi_T] 
= \mathbb{P}[\cE_{1,T} \mid \Xi_T]^m  
= \mathbb{P}[\cE\mid \Xi_T]^m .
\end{align}
Since $1 \le \tau \le n$ holds deterministically by definition~\eqref{eq:TAU}, 
after inserting~\eqref{eq:independence} into~\eqref{eq:Condition-1} it follows that 
\begin{equation}\label{eq:Use-Independence-of-Randomness}
\begin{split}
    \mathbb{P}[\cS] &= \sum_{1\le T\le n}\mathbb{P}[\tau=T, \; \cS] \\ 
    &=\sum_{1\le T\le n}\mathbb{E}\bigsqpar{\ind\{\tau=T\}\mathbb{P}[\cE\mid \Xi_T]^m}\\ 
    &=\mathbb{E}\Bigsqpar{\sum_{1\le T\le n}\ind\{\tau=T\}\mathbb{P}[\cE\mid \Xi_T]^m} . 
\end{split}
\end{equation}
We note that conditionally on 
$\Xi_1, \ldots, \Xi_n$ we have
\begin{align*}
\sum_{1\le T\le n}\ind\{\tau=T\}\mathbb{P}[\cE\mid \Xi_T]^m
=
\Bigpar{\sum_{1\le T\le n}\ind\{\tau=T\}\mathbb{P}[\cE\mid \Xi_T]}^m,
\end{align*}
which, by taking expectations, implies together with~\eqref{eq:Use-Independence-of-Randomness} that  
\begin{align}\label{eq:PrcS:equality}
\Pr[\cS] 
 = \EE\biggsqpar{\Bigpar{\sum_{1 \le T \le n}\indic{\tau=T}\Pr[\cE\mid \Xi_T]}^m} .
\end{align}
Using that~$x \mapsto x^m$ is convex for every~$x \ge 0$ (due to~$m \ge 1$), by applying Jensen's inequality to the right-hand side of~\eqref{eq:PrcS:equality} it follows~that 
\begin{equation}\label{eq:lower:Pr:cS}
\begin{split}
\Pr[\cS] 
 \ge \biggpar{\EE\Bigsqpar{\sum_{1 \le T \le n}\indic{\tau=T}\Pr[\cE\mid \Xi_T]}}^m .
\end{split}
\end{equation}
Since $1 \le \tau \le n$ holds deterministically by definition~\eqref{eq:TAU}, we also have 
\begin{equation}\label{eq:identity:E}
\begin{split}
\EE\Bigsqpar{\sum_{1 \le T \le n}\ind\{\tau=T\}\Pr[\cE\mid \Xi_T]}  = \sum_{1 \le T \le n} \EE\bigsqpar{\ind\{\tau=T\}\Pr[\cE \mid \Xi_T]} = \sum_{1 \le T \le n} \Pr[\tau=T, \; \cE] = \Pr[\cE],
\end{split}
\end{equation}
where the fact that the event $\tau=T$ is determined by $\Xi_T$ implies
\[
\mathbb{E}\bigsqpar{\ind\{\tau=T\}\Pr[\cE\mid \Xi_T]} = \mathbb{E}\bigsqpar{\mathbb{P}[\tau=T, \; \cE| \Xi_T]} = \mathbb{P}[\tau=T, \; \cE].
\]
Finally, combining~\eqref{eq:lower:Pr:cS} with~\eqref{eq:identity:E}
completes the proof of estimate~\eqref{eq:E-s-A1-m} from Lemma~\ref{lem:comparison}.
\end{proof}

\subsubsection{Upper bound on $\Pr[\cS]$: unlikely independent sets}\label{sec:success:upper}
To estimate the `success' probability $\Pr[\cS]$ defined in Section~\ref{sec:success} from above, 
the idea is to show that the event~$\cS$ implies the existence of a certain 
family of `overlapping' independent sets, 
which intuitively only exist with very small probability. 
The following definition and the subsequent lemmas formalize this~idea, 
where~$V_T = \{v_1, \ldots, v_T\}$ denotes the first $T$ vertices considered by the algorithm, 
and~$E_T = S_1\cup \cdots \cup S_T$ contains the additional vertex-pairs revealed in~(b) during the first~$T$ rounds of the~algorithm. 
\begin{definition}\label{def:Forbidden-Pattern}%
Let $N:=\lrceil{(1-\eps/2)\log_b(np)}$. 
For any integer $1 \le T \le n$, we denote by $X_{m,T}$ the number of $m$-tuples of sets $(I_1,\dots,I_m)$ with~$I_1,\dots,I_m\subseteq [n]$ that satisfy the following three~properties:%
{\vspace{-0.25em}%
\begin{romenumerate}
\itemsep 0em \parskip 0.125em  \partopsep=0.125em \parsep 0.125em  
    \item\label{def:Forbidden-Pattern:a} for every $1\le i\le m$, the set $I_i$ is an independent set in $G_i^{(T)}$, with $|I_i|\ge (1+\epsilon)\log_b(np)$, 
    \item\label{def:Forbidden-Pattern:c} we have ${I_i\cap V_T} = {I_1\cap V_T}$ for every $1\le i\le m$, with $|I_1\cap V_T|=N$, and 
		\item\label{def:Forbidden-Pattern:d} we have $\bigl|\binom{I_i}{2} \cap E_T\bigr| \le c \log_b^2(np)$ for every $1\le i\le m$.
\end{romenumerate}}%
\end{definition}
\begin{lemma}\label{lem:inclusion}
We have
\begin{equation}\label{eq:lem:inclusion}
\mathbb{P}[\cS] \le \sum_{1\le T\le n}\mathbb{P}[X_{m,T}\ge 1] .
\end{equation}
\end{lemma}
\begin{proof}
Since $1 \le \tau \le n$ holds deterministically by definition~\eqref{eq:TAU}, 
it follows that 
\[
\mathbb{P}[\cS] = \sum_{1\le T\le n}\mathbb{P}[\tau=T, \; \cS] .
\]
To complete the proof of estimate~\eqref{eq:lem:inclusion}, it thus suffices to show that for any~$1 \le T \le n$ we~have
\begin{equation}
\label{eq:N-m-T-1}    
\{\tau=T, \; \cS\}\subseteq \{X_{m,T}\ge 1\}.
\end{equation}
To justify this inclusion, we define 
\begin{equation}\label{def:Ii}
I_i := \A_n(G_i^{(\tau)}) \qquad \text{for all $1\le i\le m$.}
\end{equation}
On the event $\{\tau=T, \; \cS\}$, note that each~$I_i$ is an independent set in $G_i^{(T)}$, and that~$\cS$ implies 
\[
\min_{1\le i\le m}|I_i| = \min_{1\le i\le m}\bigl|\A(G_i^{(T)})\bigr| \ge (1+\epsilon)\log_b(np),
\]
so that $(I_1,\dots,I_m)$ satisfies property~\ref{def:Forbidden-Pattern:a} in Definition~\ref{def:Forbidden-Pattern}.
Observing that the algorithm (see Definition~\ref{def:Operation}) satisfies $\A_n(G_1^{(T)}) \cap V_T = \A_T(G_1^{(T)})$, 
on the event $\{\tau=T\}$ it follows using the definitions~\eqref{eq:G_1T} and~\eqref{eq:TAU} of $G_1^{(T)}=G_1$ and~$\tau$~that 
\[
|I_1 \cap V_T| = \bigl|\A_T(G_1^{(T)})\bigr|=|\A_T(G_1)| =\lrceil{(1-\eps/2)\log_b(np)} =N.
\]
Since similarly $\A_n(G_i^{(T)}) \cap V_T = \A_T(G_i^{(T)})$, on the event $\{\tau=T\}$ 
it follows using Remark~\ref{rem:run}~that 
\[
I_i \cap V_T = \A_T(G_i^{(T)}) =\A_T(G_1^{(T)}) = I_1 \cap V_T \qquad \text{for all $1\le i\le m$,}
\]
so that $(I_1,\dots,I_m)$ satisfies property~\ref{def:Forbidden-Pattern:c} in Definition~\ref{def:Forbidden-Pattern}.
Finally, using~$E_T \subseteq E_n$, the $c$-restricted algorithm~$\cA$ ensures by~\eqref{eq:CondInformal} that
\[
\bigl|\tbinom{I_i}{2} \cap E_T\bigr| 
 \le \bigl|\tbinom{\A_n(G_i^{(\tau)})}{2} \cap E_n\bigr| \le c \log^2_b(np),
\]
so that $(I_1,\dots,I_m)$ satisfies property~\ref{def:Forbidden-Pattern:d} in Definition~\ref{def:Forbidden-Pattern}.
To sum up, on the event $\{\tau=T, \; \cS\}$ we showed that $(I_1,\dots,I_m)$ 
satisfies properties~\ref{def:Forbidden-Pattern:a}--\ref{def:Forbidden-Pattern:d} in Definition~\ref{def:Forbidden-Pattern}, 
which by definition of~$X_{m,T}$ establishes the desired inclusion~\eqref{eq:N-m-T-1} and thus completes the proof of Lemma~\ref{lem:inclusion} (as~discussed).
\end{proof}

The following proposition intuitively states that the right-hand side of~\eqref{eq:lem:inclusion} and thus the success probability $\Pr[\cS]$ is very small, provided that~$d$ and~$n$ are sufficiently large.  
We defer the proof of Proposition~\ref{prop:MAIN} to Section~\ref{pf:MAIN-OGP}, 
since it is rather tangential to the main argument here. 
\begin{proposition}\label{prop:MAIN}
There exist constants $d,n_1 >1$ (which may depend on~$\epsilon,c,m$)  
such that, for all $n \ge n_1$, any $c$-restricted algorithm $\A$ satisfies
\begin{equation}\label{eq:prop:main}
		\sum_{1\le T\le n}
		\mathbb{P}[X_{m,T} \ge 1] \le (np)^{-\log_b(np)/2}
\end{equation}
for all edge-probabilities $p=p(n)$ satisfying $d/n \le p \le 1-n^{-1/d}$.
\end{proposition}

\subsection{Completing the proof of Theorem~\ref{thm:MAIN}}\label{sec:put:together} 
Assuming Proposition~\ref{prop:MAIN}, we are now ready to complete the proof of Theorem~\ref{thm:MAIN}. 
Indeed, combining estimate~\eqref{eq:E-s-A1-m} from Lemma~\ref{lem:comparison} with estimates~\eqref{eq:lem:inclusion} and~\eqref{eq:prop:main} from Lemma~\ref{lem:inclusion} and Proposition~\ref{prop:MAIN}, using the definitions~\eqref{eq:p-c:m} of~$m=\ceil{16/\eps^2}$ and~$\xi = \eps^2/64$ it follows (for sufficiently large~$d$ and~$n$) that, say, 
\begin{equation}\label{eq:comparison}
\Pr[\cE] \le \Pr\bigl[\cS\bigr]^{1/m} \le (np)^{-\log_b(np)/(2m)} \le (np)^{-\xi\log_b(np)}.
\end{equation}
Since~$\cE$ denotes the event that ${\bigl|\cA(G_1)\bigr|} \ge {(1+\epsilon)\log_b(np)}$, using that $G_1 \sim \G(n,p)$ by~\eqref{eq:G_1} we also have
\begin{equation}\label{eq:comparison:2}
\Pr[\cE] = \mathbb{P}_{G\sim \G(n,p)}\bigl[|\A(G)|\ge (1+\epsilon)\log_b(np)\bigr] . 
\end{equation}
Combining~\eqref{eq:comparison}--\eqref{eq:comparison:2} with~\eqref{eq:deterministic-A} therefore establishes the desired probability bound~\eqref{eq:thm:MAIN}. 
To complete the proof of Theorem~\ref{thm:MAIN}, it remains to prove Proposition~\ref{prop:MAIN}, which is the subject of the next~section.

\section{Unlikely independent sets: proof of Proposition~\ref{prop:MAIN}}\label{pf:MAIN-OGP}
In this section we give the deferred proof of Proposition~\ref{prop:MAIN}, 
which shows that a certain family of `overlapping' independent sets only exists with very small probability. 
To avoid clutter, we set 
\begin{equation}\label{def:L}
L :=\log_b(np), 
\qquad \ell := \log(np), 
\qquad \gamma:=1-\epsilon/2, 
\quad \text{ and } \quad N:=\lrceil{\gamma L},
\end{equation}
and henceforth always tacitly assume that $n$ is sufficiently large whenever necessary (how large may depend on $\epsilon,c,m,d$).

We claim that, to complete the proof of Proposition~\ref{prop:MAIN},  
it suffices to show that there is a constant~${d_0 >1}$ (which may depend on~$\epsilon,c,m$)
such that, for all $d \ge d_0$, we have
\begin{align}
\label{eq:prop:main:large}
\max_{1 \le T \le n}\Pr[Y_{m,T} \ge 1] & \;\le\; m \cdot (np)^{-L},\\
\label{eq:prop:main:typical}
\max_{1 \le T \le n}\Pr[Z_{m,T} \ge 1] & \;\le\; n^{m} \cdot (np)^{-L},
\end{align}
for all edge-probabilities $p=p(n)$ satisfying $d/n \le p \le 1-n^{-1/d}$, 
where $Y_{m,T}$ and $Z_{m,T}$ count the number of $m$-tuples $(I_1,\dots,I_m)$ in Definition~\ref{def:Forbidden-Pattern} with ${\max_{1\le i\le m}|I_i|} > {3L}$ and  ${\max_{1\le i\le m}|I_i|} \le {3L}$, respectively. 
To see this claim, note that by combining $X_{m,T}=Y_{m,T} + Z_{m,T}$ with estimates~\eqref{eq:prop:main:large}--\eqref{eq:prop:main:typical}, 
using the definition~\eqref{def:L} of~$L=\log_b(np)$ it follows (for sufficiently large~$d$ and~$n$) that, say, 
\begin{equation}
\label{eq:prop:main:proof:1}
\begin{split}
		\sum_{1\le T\le n} \mathbb{P}[X_{m,T} \ge 1] 
		& \le \sum_{1\le T\le n} \Bigpar{\Pr[Y_{m,T} \ge 1]  + \Pr[Z_{m,T} \ge 1]}  \\
		& \le n \cdot 2mn^{m} \cdot (np)^{-\log_b(np)} .
\end{split}
\end{equation}
Using the lower bound~\eqref{eq:logbnp:estimate} for~$\log_b(np)$, 
by distinguishing the cases~$p \ge n^{-1/2}$ and $d/n \le p \le n^{-1/2}$, say, 
it is routine to see that for sufficiently large $d$ and $n$ (depending on $m,c_0$) we have 
\begin{equation}\label{eq:prop:main:proof:2}
(np)^{-\log_b(np)/2} \le (np)^{- c_0\log(d)/(2p)} \le \frac{1}{2mn^{m+1}},
\end{equation}
which together with~\eqref{eq:prop:main:proof:1} readily proves estimate~\eqref{eq:prop:main} and thus Proposition~\ref{prop:MAIN}, as~claimed. 

The remainder of this section is devoted to the proof of the probability estimates~\eqref{eq:prop:main:large} and~\eqref{eq:prop:main:typical}, 
for which we shall combine random graph and optimization arguments.

\subsection{Proof of estimate~\eqref{eq:prop:main:large}: large independent sets}
Fix~$1 \le T \le n$. 
Recall that $Y_{m,T}$ counts the number of $m$-tuples $(I_1,\dots,I_m)$ in Definition~\ref{def:Forbidden-Pattern} with ${\max_{1\le i\le m}|I_i|}> {3L}$. 
Note that $Y_{m,T} \ge 1$ implies existence of $1 \le i \le m$ such that $G_i^{(T)}$ contains an independent set~$I$ of size 
\[
k := \lrceil{3L},
\] 
where~$G_i^{(T)} \sim \G(n,p)$ by Remark~\ref{rem:marginals}. 
Using estimate~\eqref{eq:logbnp:estimate} and $L =\log_b(np)$, note that for sufficiently large~$d$ we have $ne/k \le np$ and $(k-1)/2 \ge 4L/3$, say, 
so that a standard union bound argument (over all possible choices of random graphs~$G_i^{(T)}$ and independent sets~$I$) and the well-known binomial coefficient estimate $\binom{n}{k} \le (ne/k)^k$~yield
\begin{equation}\label{eq:ISET:estimate}
\begin{split}
\Pr[Y_{m,T} \ge 1]
& \le \sum_{1 \le i \le m} \sum_{I \subseteq [n]:|I|=k}(1-p)^{\binom{k}{2}} \\
& \le m \cdot \binom{n}{k}(1-p)^{\binom{k}{2}} \le m \cdot \biggpar{\frac{ne}{k} \cdot (1-p)^{(k-1)/2}}^k \\
& \le m \cdot \Bigpar{np \cdot (np)^{-4/3} }^k \le m \cdot (np)^{-L} ,
\end{split}
\end{equation}
completing the proof of estimate~\eqref{eq:prop:main:large}.

\subsection{Proof of estimate~\eqref{eq:prop:main:typical}: overlapping independent sets}
Fix~$1 \le T \le n$. 
Recall that  $Z_{m,T}$ counts the number of $m$-tuples $(I_1,\dots,I_m)$ in Definition~\ref{def:Forbidden-Pattern} with ${\max_{1\le i\le m}|I_i|} \le {3L}$. 
Given $\boldsymbol{\alpha}=(\alpha_1,\dots,\alpha_m) \in [1+\epsilon,3]^m$, we introduce the auxiliary random variable $Z_{\boldsymbol{\alpha},m,T}$ that counts the number of all $(I_1,\dots,I_m)$ arising in Definition~\ref{def:Forbidden-Pattern} with the additional constraint that $|I_i|=\alpha_i L$ for all $i \in [m]$. 
Note that $\boldsymbol{\alpha}L \in\mathbb{N}^m$ necessarily, and that in fact $\boldsymbol{\alpha}L \in [n]^m$ holds.
Hence 
\begin{equation}\label{eq:Zmt}
    Z_{m,T} = \sum_{\substack{\boldsymbol{\alpha} \in[1+\epsilon,3]^m : \\ \boldsymbol{\alpha}L\in[n]^m}} Z_{\boldsymbol{\alpha},m,T}.
\end{equation}
Using Markov's inequality and linearity of expectation, we thus infer that 
\begin{equation}\label{eq:Pr:eq:Zmt}
\Pr[Z_{m,T} \ge 1] \le \EE[Z_{m,T}]  = \sum_{\substack{\boldsymbol{\alpha} \in[1+\epsilon,3]^m: \\ \boldsymbol{\alpha}L\in[n]^m}} \EE[Z_{\boldsymbol{\alpha},m,T}].
\end{equation}
We shall bound $\EE[Z_{\boldsymbol{\alpha},m,T}]$ in the following subsections, using counting and probability estimates.

\subsubsection{Counting estimates} 
We start with the following counting estimate for the number of $m$-tuples $(I_1,\dots,I_m)$ with specific sizes. 
Recall that 
$V_T = \{v_1, \ldots, v_T\}$ denotes the first~$T$ vertices considered by the algorithm~$\cA$, 
which we assume to be known in Lemma~\ref{lemma:Count-Est} below.
\begin{lemma}\label{lemma:Count-Est}
For sufficiently large~$d$ the following is true for any $\boldsymbol{\alpha}=(\alpha_1,\dots,\alpha_m) \in[1+\epsilon,3]^m$.
Given~$V_T \subseteq [n]$,  
the total number of $(I_1,\dots,I_m)$ with~$I_1,\dots,I_m\subseteq [n]$ that satisfy property~\ref{def:Forbidden-Pattern:c} in Definition~\ref{def:Forbidden-Pattern} and $|I_i|=\alpha_iL$ for all $1 \le i \le m$ is at most
\begin{equation}\label{eq:COUNTING_TERM}
\exp\biggpar{\ell L \cdot \Bigsqpar{\sum_{1\le i \le m}\bigpar{\alpha_i-\gamma}+\gamma}}.
\end{equation}
\end{lemma}
\begin{proof}
To bound the total number of $(I_1,\dots,I_m)$, we will describe a way to construct all such tuples (and potentially more tuples), and bound the number of choices in our construction from above. 
We first choose the set $I_1 \cap V_T$ of ${N=\lrceil{\gamma L}}$ common vertices, which  
using~$V_T \subseteq [n]$ can be done in at most   
\begin{equation*}
\binom{|V_T|}{N} \le \binom{n}{\lrceil{\gamma L}} 
\end{equation*}
many ways. 
Given~$I_1 \cap V_T$ of size~$|I_1 \cap V_T|=N$, note that property~\ref{def:Forbidden-Pattern:c} in Definition~\ref{def:Forbidden-Pattern} enforces $I_i\cap V_T=I_1 \cap V_T$ for each~$1 \le i \le m$. 
With this in mind, we next choose the remaining vertices of the sets~$I_i$ of size $|I_i|=\alpha_i L$, 
which using $|I_1 \cap V_T|=N=\lrceil{\gamma L}$ can be done in at most 
\[
\prod_{1\le i\le m} \binom{n}{\alpha_i L-\lrceil{\gamma L}}  
\]
many ways.
Using the standard estimate $\binom{n}{k} \le (ne/k)^k$, 
it follows that the total number of $(I_1,\dots,I_m)$ that satisfy property~\ref{def:Forbidden-Pattern:c} in Definition~\ref{def:Forbidden-Pattern} and $|I_i|=\alpha_iL$ for all $1 \le i \le m$ is at most
\begin{equation*}
\binom{n}{\lrceil{\gamma L}} \cdot \prod_{1\le i\le m} \binom{n}{\alpha_i L-\lrceil{\gamma L}} \le \biggpar{\frac{ne}{\lrceil{\gamma L}}}^{\lrceil{\gamma L}} \cdot \prod_{1\le i\le m} \biggpar{\frac{ne}{\alpha_i L-\lrceil{\gamma L}}}^{\alpha_i L-\lrceil{\gamma L}} ,
\end{equation*}
which in turn, by noting $\min\{\gamma, \: \alpha_i-\gamma\} \ge \min\{1-\eps/2, \:  3\eps/2\} \ge \eps/2$ and using estimates similarly to~\eqref{eq:ISET:estimate}, is for sufficiently large~$d$ routinely seen to be at~most, say, 
\begin{equation*}
\begin{split}
\bigpar{np}^{\lrceil{\gamma L}} \cdot \prod_{1\le i\le m} \bigpar{np}^{\alpha_i L -\lrceil{\gamma L}} & = \bigpar{np}^{\sum_{1 \le i \le m}\alpha_i L - (m-1) \lrceil{\gamma L}},
\end{split}
\end{equation*}
which using $np = e^{\ell}$ readily completes the proof of Lemma~\ref{lemma:Count-Est}. 
\end{proof}

\subsubsection{Probability estimates}%
We next prove the following probability estimate. 
Similar as in the proof of Lemma~\ref{lem:comparison}, we henceforth let~$\Xi_T$ encode the edge-status of all vertex-pairs~$\{u,v\} \in \tbinom{V_T}{2} \cup E_T$, 
i.e., the edge-status of all vertex-pairs revealed during the first~$T$ rounds of the algorithm~$\cA$. 
Recalling that the algorithm~$\cA$ is deterministic, note that~$V_T = \{v_1, \ldots, v_T\}$ and~$E_T = S_1\cup \cdots \cup S_T$ are both determined by~$\Xi_T$, 
where~$E_T$ contains the additional future vertex-pairs revealed in~(b) during the first~$T$ rounds of~$\cA$. 
%
\begin{lemma}\label{lemma:Prob-Est}
For sufficiently large~$d$ the following is true for any $\boldsymbol{\alpha}=(\alpha_1,\dots,\alpha_m) \in[1+\epsilon,3]^m$. 
For every $m$-tuple $(I_1,\dots,I_m)$ with~$I_1,\dots,I_m\subseteq [n]$ 
that satisfies properties~\ref{def:Forbidden-Pattern:c} and~\ref{def:Forbidden-Pattern:d} in Definition~\ref{def:Forbidden-Pattern}
as well as $|I_i|=\alpha_iL$ for all~$i \in [m]$,
we have 
\begin{equation}\label{eq:lemma:Prob-Est}
\begin{split}
&\mathbb{P}\Bigl[\text{$I_i$ is an independent set in $G_i^{(T)}$ for all $1 \le i \le m$} \: \Big| \: \Xi_T\Bigr]  \\
    & \quad \le \exp\biggpar{-\ell L \cdot \sum_{1\le i\le m}\Bigpar{\alpha_i^2/2- \gamma^2/2-2c}}.
\end{split}
\end{equation}
\end{lemma}
\begin{proof}
Given a set $I_i\subseteq[n]$ as in the statement of the lemma, we define the set 
\begin{equation}\label{eq:forbidden:edges}
\begin{split}
  F_i & := \tbinom{I_i}{2} \setminus \Bigpar{\tbinom{V_T}{2} \cup E_T} ,
\end{split}
\end{equation}
which is determined by~$\Xi_T$ (since~$V_T$ and~$E_T$ are both determined by~$\Xi_T$). 
Focusing on the absence of the edges in the $F_i$ 
(ignoring the absence of other potential edges might seem wasteful, but later on we can compensate for this in~\eqref{eq:N-alpha-m-T}--\eqref{eq:PSI:2} by taking $m$ sufficiently large) 
it follows that 
\begin{align}\label{eq:E3-indep}
    \Pr\Bigl[\text{$I_i$ is an independent set in $G_i^{(T)}$ for all $1 \le i \le m$} \: | \: \Xi_T\Bigr] & \le \Pr[\cF \: | \: \Xi_T]
\end{align}
for the `forbidden edges' event
\begin{equation}\label{eq:forbidden:event}
\begin{split}
		\cF := \Bigcpar{e_{i,T}\bigpar{\{u,v\}}=0 \text{ for all $\{u,v\} \in F_i$ and $i \in [m]$}},
\end{split}
\end{equation}
where~$e_{i,T}\bigpar{\{u,v\}} \in \{0,1\}$ denotes the edge-status of the pair~$\{u,v\}$ in~$G_i^{(T)}$, see Section~\ref{sec:interpolation}.  
We next give lower bounds on the sizes of the sets~$F_i$ 
defined in~\eqref{eq:forbidden:edges}, 
exploiting that properties~\ref{def:Forbidden-Pattern:c} and~\ref{def:Forbidden-Pattern:d} in Definition~\ref{def:Forbidden-Pattern} ensure that $|I_i \cap V_T|=|I_1 \cap V_T|=N = \lrceil{ \gamma L}$ and $\bigl|\binom{I_i}{2} \cap E_T\bigr| \le c \log_b^2(np) = cL^2$. 
Namely, together with~$|I_i|= \alpha_i L$ and~$\alpha_i +\gamma \le 4$, 
it follows that we deterministically have, say, 
\begin{equation}\label{eq:Fi:bound}
\begin{split}
|F_i| & \ge \bigl|\tbinom{I_i}{2}\bigr| - \bigl|\tbinom{I_i \cap V_T}{2}\bigr| -  \bigl|\tbinom{I_i}{2} \cap E_T\bigr|\\
& \ge (\alpha_i L-1)^2/2 - (\gamma L+1)^2/2 - cL^2\\
& = L^2\Bigpar{\alpha_i^2/2 - \gamma^2/2-(\alpha_i+\gamma)/L - c}\\
& \ge L^2\Bigpar{\alpha_i^2/2 - \gamma^2/2-4/L - c}.
\end{split}
\end{equation}
Note that all vertex-pairs $\{u,v\} \in F_1 \cup \cdots \cup F_m$ satisfy~$\{u,v\} \not\in\tbinom{V_T}{2} \cup E_T$, 
which means that they have not been revealed during the first~$T$ rounds of the algorithm. 
By inspecting~\eqref{def:GiT:late} it therefore follows that, conditional on~$\Xi_T$, all the vertex-pairs in $F_{1} \cup \cdots \cup F_{m}$ are included as edges into each graph $G_1^{(T)}, \ldots, G_m^{(T)}$ independently with probability~$p$.  
In view of~\eqref{eq:forbidden:event} and the fact that~$F_i$ is determined by~$\Xi_T$, 
using the deterministic lower bound~\eqref{eq:Fi:bound} on~$|F_i|$ and $L=\log_b(np) = -\ell/\log(1-p)$, 
it thus follows that 
\begin{equation*}
\begin{split}
    \Pr[\cF \: | \: \Xi_T] 
		& = \prod_{1 \le i \le m}(1-p)^{|F_i|}\\
		& \le \prod_{1 \le i \le m}(1-p)^{L^2\bigpar{\alpha_i^2/2 - \gamma^2/2-4/L-c}} \\
		& = \exp\biggpar{-\ell L \cdot \sum_{1\le i\le m}\Bigpar{\alpha_i^2/2- \gamma^2/2-4/L-c}}.
\end{split}
\end{equation*}
Recalling that~$L=\log_b(np) \ge c_0 \log(d)$ by~\eqref{eq:logbnp:estimate}, for sufficiently large~$d$ (depending on~$c$) we infer that~${4/L \le c}$, 
which together with~\eqref{eq:E3-indep} completes the proof of estimate~\eqref{eq:lemma:Prob-Est} and thus  Lemma~\ref{lemma:Prob-Est}.
\end{proof}

\subsubsection{Completing the proof: putting everything together}
After these preparations we are now ready to bound $\mathbb{E}[Z_{\boldsymbol{\alpha},m,T}]$, 
by combining the number of possible~$(I_1,\dots,I_m)$ via the counting estimate from Lemma~\ref{lemma:Count-Est}, 
with the probability estimate from Lemma~\ref{lemma:Prob-Est}. 
Indeed, since the vertex set~$V_T$ is determined by~$\Xi_T$ (as mentioned above Lemma~\ref{lemma:Prob-Est}), 
by multiplying estimates~\eqref{eq:COUNTING_TERM} and~\eqref{eq:lemma:Prob-Est} it readily follows for any $\boldsymbol{\alpha}=(\alpha_1,\dots,\alpha_m) \in[1+\epsilon,3]^m$ with $\boldsymbol{\alpha}L \in [n]^m$ that 
\begin{equation}\label{eq:N-alpha-m-T}
    \EE[Z_{\boldsymbol{\alpha},m,T}]
		= \EE\bigsqpar{\EE[Z_{\boldsymbol{\alpha},m,T} \: | \: \Xi_T]} 
		\le \EE\Bigsqpar{\exp\Bigpar{-\ell L \cdot \Psi(\boldsymbol{\alpha})}}
		= \exp\Bigpar{-\ell L \cdot \Psi(\boldsymbol{\alpha})}
\end{equation}
for the function
\begin{equation}\label{eq:PSI}
\Psi(\boldsymbol{\alpha}) := \sum_{1\le i\le m}\Bigsqpar{\bigpar{\alpha_i^2/2-\alpha_i}-\bigpar{\gamma^2/2-\gamma}-2c} - \gamma.
\end{equation}
We now bound $\Psi(\boldsymbol{\alpha})$ from below.
Recalling that~$\gamma=1-\epsilon/2$, we have 
\[
\gamma^2/2-\gamma = \eps^2/8-1/2.
\]
Furthermore, by noting that the following minimum is attained for $\alpha_i=1+\epsilon$, we have 
\[
\min_{\alpha_i \in [1+\epsilon,3]}\bigpar{\alpha_i^2/2-\alpha_i} = \eps^2/2-1/2 .
\]
Inserting these two estimates into~\eqref{eq:PSI}, 
we now exploit that by definition~\eqref{eq:p-c:m} we have $2c \le \eps^2/4$ and $m \ge 16/\eps^2$, 
so that using~$\gamma \le 1$ we infer that  
\begin{equation}\label{eq:PSI:2}
\begin{split}
\Psi(\boldsymbol{\alpha}) & \ge m \cdot \bigpar{\eps^2/2-\eps^2/8- \eps^2/4} - \gamma 
 \ge m \cdot \eps^2/8 -1 \ge 1 .
\end{split}
\end{equation}

To sum up, since there are at most~$n^m$ choices of $\boldsymbol{\alpha}\in[1+\epsilon,3]^m$ for which $\boldsymbol{\alpha}L \in [n]^m$, 
by inserting~\eqref{eq:N-alpha-m-T} and~\eqref{eq:PSI:2} into~\eqref{eq:Pr:eq:Zmt} it follows that 
\begin{align*}
\Pr[Z_{m,T} \ge 1] \le \sum_{\substack{\boldsymbol{\alpha} \in[1+\epsilon,3]^m: \\ \boldsymbol{\alpha}L\in[n]^m}} \EE[Z_{\boldsymbol{\alpha},m,T}] \le n^m \cdot \exp\bigpar{-\ell L},
\end{align*}
which using~$\ell=\log(np)$ completes the proof of estimate~\eqref{eq:prop:main:typical} and thus Proposition~\ref{prop:MAIN}, as discussed.

\begin{remark}
One can check that the proofs of Proposition~\ref{prop:MAIN} and Theorem~\ref{thm:MAIN} remain valid when the fairly natural assumption ${d/n \le p \le 1-n^{-1/d}}$ 
is replaced by the more technical~assumption 
\[
C_0 \max\biggcpar{\frac{1}{Lp}, \: \sqrt{\frac{\max\{1, \log(n)/\ell\}}{L}}} \le \eps \le 1 ,
\]
where~$C_0> 0$ is a sufficiently large universal constant (that does not depend on~$\eps,c,\xi,m,n_0,\gamma$), 
and the parameters $L =\log_b(np)$ and $\ell = \log(np)$ are defined as in~\eqref{def:L}. 
In particular, for constant~$p \in (0,1)$ it follows that Theorem~\ref{thm:MAIN} remains valid when $2C_0/\sqrt{\log_b(np)} \le \eps \le 1$, say. 
\end{remark}

\section{Achievability result: proof of Theorem~\ref{thm:algorithm}}\label{sec:PF-Greedy-BruteForce}
In this section we prove our achievability result Theorem~\ref{thm:algorithm}, 
which shows that a certain online algorithm can find an independent set of size $(1 + \eps)\log_b(n)$ with high probability. 
Heuristically speaking, the algorithm we shall use consists of three phases: 
{\vspace{-0.25em}%
\begin{romenumerate}
\itemsep 0em \parskip 0.125em  \partopsep=0.125em \parsep 0.125em  
    \item\label{heur:algo:phase1} first we greedily find an independent set~$I_T$ of size~$|I_T| \approx(1-\eps)\log_b(n)$, 
    \item\label{heur:algo:phase2} then we iteratively find a vertex set~$R$ of size~$|R| \approx n^{\eps}$ that only contains vertices $w \not\in I_T$ which have no neighbors in~$I_T$, and 
    \item\label{heur:algo:phase3} finally we find using brute-force search an independent set $J \subseteq R$ of size approximately $|J| \approx 2\log_b(|R|) \approx 2\eps \log_b(n)$, 
		and construct an independent set~$I:=I_T \cup J$ of size~$|I| \approx (1+\eps)\log_b(n)$.
\end{romenumerate}}%

Turning to the details, we now describe a deterministic algorithm~$\cA$ that implements the above-mentioned three-phase idea in a way that fits it into the online-framework of Definition~\ref{def:Operation}. Setting 
\begin{equation}\label{def:parameters}
\ell:= \bigfloor{(1-\epsilon)\log_b(n) - 3.5 \log_b\log_b(n)},  
\qquad  
T := \ell \cdot \lrceil{n^{1-\eps} \ell},
\quad \text{ and } \quad 
r:= \lrfloor{n^{\eps}\log^3_b(n)} ,
\end{equation}
the algorithm~$\cA$ initially starts with $I_0:=\emptyset$, and then proceeds in rounds as follows.%
{\vspace{-0.25em}
\begin{itemize}
\itemsep 0em \parskip 0.125em  \partopsep=0.125em \parsep 0.125em  
    \item {\bf Greedy rounds:} In round $1 \le t \le T$ we consider vertex~$t \in [n]$, 
		and reveal the edge-status of all vertex-pairs~$\{u,t\}$ with~$u \in I_{t-1}$. 
    If either~$|I_{t-1}| \ge \ell$ or vertex~$t$ has a neighbor in $I_{t-1}$, then we set $I_t := I_{t-1}$.
		Otherwise, we set $I_t := I_{t-1}\cup \{t\}$.
    \item {\bf Search round:} In round $t=T+1$, we consider vertex~$n$ and set $I_{T+1} := I_{T}$. 
		In this round we use the set~$S_{T+1}$ to reveal the edge-status of all vertex-pairs $\{u,v\}$ with $u \in I_T$ and $v \in \{T+1, \ldots, n-2\}=:W$, 
		in order to determine the set~$R \subseteq W$ of vertices in~$W$ that have no neighbors in~$I_T$. 
		If~$|R| > r$, then we iteratively remove vertices from~$R$ (in lexicographic order, say) until the resulting set satisfies $|R| =r$. 
    \item {\bf Brute-force round:} 
		In round $t=T+2$, we consider vertex~$n-1$ and set $I_{T+2} := I_{T+1}$. 
		In this round we use the set~$S_{T+2}$ to reveal the edge-status of all vertex-pairs ${\{u,v\} \in \tbinom{R}{2}}$, 
		and then use brute-force search to find a largest independent set~$J \subseteq R$ of size at most~${|J| \le \lrfloor{2 \log_b(r)}}$, where we write~${J=\{w_1, \ldots, w_{|J|}\}}$. 
    \item {\bf Final rounds:} 
		In rounds $T+3 \le t \le T+2+|J|$, we consider vertex~$w_{t-(T+2)} \in [n]$ and set $I_{t} := I_{t-1} \cup \{w_{t-(T+2)}\}$.
		In the remaining rounds $T+3+|J| \le t \le n$ we sequentially consider the remaining vertices in $W \setminus J$, and each time set $I_{t} := I_{t-1}$.
\end{itemize}}%
\noindent
Note that algorithm~$\cA$ has running time \[
n^{\Theta(1)} + \tbinom{r}{\lrfloor{2 \log_b(r)}} \cdot n^{\Theta(1)} = n^{(2+o(1))\eps^2 \log_b(n)},
\]
where here and below we always tacitly assume that $n$ is sufficiently large whenever necessary (how large may depend on~${\eps,p}$). 
Furthermore, 
we have~$I_n = I_T \cup J$ and $E_n = {S_1\cup \cdots \cup S_n}={S_{T+1} \cup S_{T+2}}$, 
so using $|J| \le \lrfloor{2 \log_b(r)}$ and $|I_T| \le \ell$ the algorithm~$\cA$ deterministically satisfies, say, 
\begin{equation}
\begin{split}
\bigl|\tbinom{I_n}{2} \cap E_n\bigr| 
& \le |J| \cdot |I_T| + \tbinom{|J|}{2}\\
& \le 2 \log_b(r) \cdot \bigsqpar{\ell + \log_b(r)} \\
& \le \bigsqpar{2\eps\log_b(n) + 6\log_b\log_b(n)} \cdot \bigsqpar{\log_b(n)-0.5\log_b\log_b(n)} \\
& \le \biggpar{2\eps + \frac{6\log_b \log_b(n)}{\log_b(n)}} \log^2_b(np) \le 3\eps \log^2_b(np),
\end{split}
\end{equation}
so that~$\cA$ is $3\eps$-restricted according to~\eqref{eq:CondInformal}. 
(The total number $|E_n| \le n \cdot |I_T|+ \tbinom{|R|}{2} \le n^{\min\{1,2\eps\}+o(1)}$ of queried future vertex-pairs is larger.) 
Note that the algorithm~$\cA$ finds an independent set of size~$|I_n|=|I_T|+|J|$. 
To complete the proof of Theorem~\ref{thm:algorithm}, it thus suffices to prove the following~lemma. 
\begin{lemma}\label{lem:bruteforce}
We have $\Pr[|I_T|+|J| \ge (1+\eps)\log_b(n)] \to 1$ as~$n \to \infty$.
\end{lemma}
The proof of Lemma~\ref{lem:bruteforce} uses the following auxiliary result, 
whose proof we defer.  
\begin{lemma}\label{lem:properties}
We have $\Pr[|I_T|=\ell, \: |R|=r] \to 1$ as~$n \to \infty$.
\end{lemma}
\begin{proof}[Proof of Lemma~\ref{lem:bruteforce}]
We henceforth condition on the outcome of the greedy rounds and the search round,
where by Lemma~\ref{lem:properties} we may assume that $|I_T|=\ell$ and~$|R|=r$. 
Since all vertex-pairs in~$\binom{R}{2}$ are included as edges independently with probability~$p$ (since their status has not yet been revealed by the algorithm), 
it follows that~$G[R] \sim \mathbb{G}(r,p)$. 
Using well-known random graph results (see~\cite[Theorem~7.1]{JLR}, for example) on the size of the largest independent of $\mathbb{G}(r,p)$,  
as~$n \to \infty$ and thus~$r \to \infty$ it follows that
\[
\Pr\bigsqpar{\alpha(G[R]) \ge 2\log_b(r)-2\log_b \log_b(rp) +2\log_b(e/2)+1-1/p} \to 1,
\]
where we omitted the conditioning from our notation to avoid clutter. 
Using~$r=\lrfloor{n^{\eps}\log^3_b(n)}$ and the definition~\eqref{def:parameters} of~$\ell$, 
with probability tending to one as~$n \to \infty$, we thus have 
\begin{equation}
\begin{split}
|I_T|+|J| & = \ell + \min\bigcpar{\alpha(G[R]), \: \lrfloor{2\log_b(r)}} \\
& \ge \bigsqpar{(1-\epsilon)\log_b(n) - 3.5 \log_b\log_b(n) - 1} + \bigsqpar{2\eps\log_b(n) + 4 \log_b \log_b(n) - O(1)}\\
& > (1+\epsilon)\log_b(n) ,
\end{split}
\end{equation}
which completes the proof of Lemma~\ref{lem:bruteforce}. 
\end{proof}
Finally, we give the deferred proof of the auxiliary result Lemma~\ref{lem:properties}.
\begin{proof}[Proof of Lemma~\ref{lem:properties}]
We start by showing that $\Pr(|I_T|=\ell) \to 1$,
for which our argument is based on the analysis of Grimmett and McDiarmid~\cite{grimmett1975colouring} from~1975.
Let $\Delta_i$ denote the number of rounds required for the independent set to grow from size~$i-1$ to~$i$.
Note that a union bound argument gives 
\begin{equation*}
\begin{split}
\Pr[|I_T| < \ell]  & = \Pr[\Delta_1+\cdots+\Delta_{\ell}> T] \\
& \le \Pr[\Delta_i > T/\ell \text{ for some $1 \le i \le \ell$}] \\
& \le \sum_{1 \le i \le \ell }\Pr[\Delta_i > T/\ell].
\end{split}
\end{equation*}
Since in round~$t$ all vertex-pairs in~$I_{t-1} \times \{t\}$ are included as edges independently with probability~$p$ (since their status has not yet been revealed by the algorithm),
using the standard estimate~${1-x \le e^{-x}}$ it follows~that 
\begin{equation*}
\begin{split}
\Pr[\Delta_i > T/\ell] \le \bigpar{1-(1-p)^{i-1}}^{T/\ell} \le e^{- (1-p)^{i-1} T/\ell}.
\end{split}
\end{equation*}
Using~$i-1 \le \ell$ and the definitions~\eqref{def:parameters} of~$\ell$ and~$T$, 
it readily follows that 
\begin{equation*}
\begin{split}
\Pr[|I_T| < \ell] \le \ell \cdot e^{- (1-p)^{\ell} T/\ell} \le \ell e^{-\ell} \to 0,
\end{split}
\end{equation*}
which together with the deterministic bound $|I_T| \le \ell$ (guaranteed by construction of the algorithm) 
shows that $\Pr(|I_T|=\ell) \to 1$, 

We henceforth condition on the outcome of the greedy rounds, and assume that~$|I_T|=\ell$ by the analysis above. 
In this conditional space we will show that $\Pr(|R|=r) \to 1$, where we here and below omit the conditioning from our notation (to avoid clutter). 
Let $X$ denote the number of vertices in~$W$ that have no neighbors in~$I_T$. 
Since all vertex-pairs in~$W \times I_T$ are included as edges independently with probability~$p$ (since their status has not yet been revealed by algorithm), 
we have
\[
X \sim {\rm Binomial}\Bigpar{|W|, \: (1-p)^{\ell}}.
\]
Recalling that~$|W|=n-(T+2)$, by definition~\eqref{def:parameters} of~$\ell$ and~$r$ we thus have, say, 
\[
\EE[X] = |W| \cdot (1-p)^{\ell} 
\ge (1-o(1)) \cdot  n^{\eps}\log_b^{3.5}(n) > 2r .
\]
Using standard  Chernoff bounds (see~\cite[Theorem~2.1]{JLR}), it follows that 
\[
\Pr[|R| < r] = \Pr[X < r] \le \Pr\bigsqpar{X \le \EE[X]/2} \le e^{-\EE[X]/8} \le e^{-r/4}\to 0 ,
\]
which together with the deterministic bound $|R| \le r$ (guaranteed by construction of the algorithm) 
shows that $\Pr[|R|=r] \to 1$, completing the proof of Lemma~\ref{lem:properties}. 
\end{proof}

\begin{remark}
The above proof of Theorem~\ref{thm:algorithm} carries over when $\eps \gg \log\log(n)/\log(n)$, 
with running-time $n^{2\eps^2 \log_b(n)+O(1)}$. 
In particular, given any constant~$C > 0$, for $\eps := C/\sqrt{\log_b(n)}$ it follows that 
the above-described online algorithm~$\cA$ whp finds an independent set $I$ in~$\mathbb{G}(n,p)$ of size 
\[
|I| \ge {\log_{b}(np) + C \sqrt{\log_{b}(np)}}
\]
using running-time $n^{2C^2 + O(1)}$, i.e., in polynomial time. 
\end{remark}

\bibliographystyle{amsalpha}
\bibliography{bib2}

\begin{appendix}

\section{Appendix}
%
\subsection{Proof Sketch of Proposition~\ref{Prop:Unstable}}\label{apx:Prop:Unstable}
\begin{proof}[Proof Sketch of Proposition~\ref{Prop:Unstable}]
Each of the $\binom{n}{2}$ edges is included independently into ${G_1\sim \G(n,p)}$ with probability~$p$.
We construct~$G_2$ as a copy of~$G_1$, 
with the wrinkle that we resample the status of the vertex-pair~${1,2}$ (which is independently included as an edge  with probability~$p$); note that $G_2\sim \G(n,p)$. 
Let $I_1$ and $I_2$ denote the independent sets returned by Algorithm~\ref{Alg:Greedy} run on~$G_1$ and~$G_2$. Note that
\[
\mathbb{P}\bigl[2 \notin I_1, \: 2\in I_2\bigr] =  \mathbb{P}\bigl[\{1,2\} \in E(G_1), \: \{1,2\} \not\in E(G_2)\bigr] =  p(1-p) = \Omega(1).
\]
Note that conditional on the event $\{2\notin I_1, \: 2\in I_2\}$, vertex~$2 \in I_2$ has $1+{\rm Binomial}(|I_1|-1,p)$ neighbors in~$I_1$, which is with high probability at least $(1-\eps)^2p\log_b(n)$, say.
Since none of these neighbors are contained in~$I_2$, this shows that, by simply resampling one edge of the random graph~$\G(n,p)$, the independent sets returned by Algorithm~\ref{Alg:Greedy} can differ substantially with probability $\Omega(1)$.
\end{proof}

\subsection{Proof of Lemma~\ref{lemma:Logb-LowerBd}}\label{apx:lemma:Logb-LowerBd}

\begin{proof}[Proof of Lemma~\ref{lemma:Logb-LowerBd}]   
Note that $-\log(1-p) = p + O(p^2)$ for~$p \le 1/2$, say. Choosing $0<p_0\le 1/2$ sufficiently small, we thus obtain that for $d/n\le p\le p_0$, we have $-\log(1-p) \le 9p$ and $\log(np) \ge \log(d)$. Consequently,  $\log_b(np) \ge \log (d)/(9p)$ in~\eqref{eq:logbnp:estimate}.

In the remaining case, $p_0 \le p \le 1-n^{-1/d}$, the crux is that $-\log(1-p) \le \log(n)/d$ and $\log(np) \ge \log(n p_0) \ge \log(n)/2$ for sufficiently large~$n$ (depending only on~$p_0$), 
so that $\log_b(np) \ge d/2 \ge p_0 d/(2p)$  in~\eqref{eq:logbnp:estimate}.
This establishes~\eqref{eq:logbnp:estimate} by checking that $\min\{\log(d), \: d\} = \log(d)$ for all~$d \ge 1$.
\end{proof}

\end{appendix}

\end{document}